\documentclass[format=acmlarge, printfolios=true, printccs=false, nonacm=true, review=false]{acmart}
\usepackage[T1]{fontenc}

\usepackage[utf8x]{inputenc}
\usepackage{array}

\usepackage{booktabs}
\newcolumntype{N}{>{\centering\arraybackslash}m{2.5in}}

\usepackage{graphicx}
\usepackage[colorinlistoftodos]{todonotes}
\usepackage{booktabs}
\usepackage{multicol}
\usepackage{hyperref}
\usepackage{tabularx}
\usepackage{subcaption}
\usepackage{multirow}

\newcommand{\FMR}{\ensuremath{{\tt FMR}}}
\newcommand{\FAR}{\ensuremath{{\tt FAR}}}
\newcommand{\FNMR}{\ensuremath{{\tt FNMR}}}
\newcommand{\FRR}{\ensuremath{{\tt FRR}}}
\newcommand{\FTA}{\ensuremath{{\tt FTA}}}
\newcommand{\EER}{\ensuremath{{\tt EER}}}
\newcommand{\FNIR}{\ensuremath{{\tt FNIR}}}
\newcommand{\FPIR}{\ensuremath{{\tt FPIR}}}

\usepackage{etoolbox}

\begin{document}
\title{Untargeted Near-collision Attacks on Biometrics: Real-world Bounds and Theoretical Limits}

\author{Axel DURBET}
\affiliation{%
    \institution{Université Clermont-Auvergne, CNRS,
        Mines de Saint-Étienne, Clermont-Auvergne-INP,
        LIMOS\country{France}}
}
\email{axel.durbet@uca.fr}

\author{Paul-Marie GROLLEMUND}
\affiliation{%
    \institution{Université Clermont-Auvergne, CNRS,
        Mines de Saint-Étienne, Clermont-Auvergne-INP,
        LMBP\country{France}}
}
\email{paul_marie.grollemund@uca.fr}

\author{Kevin THIRY-ATIGHEHCHI}
\affiliation{%
    \institution{Université Clermont-Auvergne, CNRS,
        Mines de Saint-Étienne, Clermont-Auvergne-INP,
        LIMOS\country{France}}
}
\email{kevin.atighehchi@uca.fr}

\begin{abstract}
    A biometric recognition system can operate in two distinct modes: \textit{identification} or \textit{verification}. In the first mode, the system recognizes an individual by searching the enrolled templates of all the users for a match. In the second mode, the system validates a user’s identity claim by comparing the fresh provided template with the enrolled template. The biometric transformation schemes usually produce binary templates that are better handled by cryptographic schemes, and the comparison is based on a distance that leaks information about the similarities between two biometric templates. Both the experimentally determined false match rate and false non-match rate through recognition threshold adjustment define the recognition accuracy, and hence the security of the system. To our knowledge, few works provide a formal treatment of security in case of minimal information leakage, \textit{i.e.}, the binary outcome of a comparison with a threshold. In this paper, we focus on untargeted attacks that can be carried out both \textit{online} and \textit{offline}, and in both \textit{identification} and \textit{verification} modes.

    On the first hand, we focus our analysis on the accuracy metrics of biometric systems. We provide the complexity of untargeted attacks using the False Match Rate (\FMR) and the False Positive Identification Rate (\FPIR) to address the security of these systems. Studying near-collisions with these metrics allows us to estimate the maximum number of users in a database, given a chosen $\FMR$, to preserve the security and the accuracy. These results are evaluated on systems from the literature.

    On the other hand, we rely on probabilistic modelling to assess the  theoretical security limits of biometric systems. The study of this metric space, and system parameters (template size, threshold and database size), gives us the complexity of untargeted attacks and the probability of a near-collision.
\end{abstract}

\keywords{Biometric authentication, Biometric identification, Biometric transformations, Cancellable biometrics, Near-collisions, Master templates}

\maketitle

\tableofcontents

\section{Introduction}

Biometric technologies provide an effective means of authentication or identification based on physical or behavioral characteristics.
Due to their convenience and speed, the use of such systems continues to grow, replace or complement the traditional password.
In a biometric recognition system, biometric
templates of users are stored in a database.
The first operating mode consists in determining the identity of an
individual
by comparing her fresh provided template with all the templates stored in the database.
The second one, the traditional authentication (verification) mode, corresponds
to the verification of the claimed identity by comparing the corresponding
enrolled template with the fresh template the individual provides.
As a consequence, service providers
need to manage biometric databases.
Biometric data are as prone to exhaustive search attacks as passwords, but unlike passwords, they cannot be efficiently revoked. Thus, biometric databases are prime targets for cyberattackers.
Biometric data are categorized as highly sensitive personal data covered by the GDPR.
Moreover, these data may disclose information like genetic information~\cite{penrose1965dermatoglyphic}
and diseases~\cite{imran2019comparative,ross2022deducing}.
The essential
security and performance criteria that must be fulfilled by biometric
recognition systems are identified in ISO/IEC 24745~\cite{ISO24745} and ISO/IEC
30136~\cite{ISO30136}: \textit{irreversibility}, \textit{unlinkability},
\textit{revocability} and \textit{performance preservation}.

Biometric templates are generated from biometric measurements (\textit{e.g.}, a
face or fingerprint image). These measurements undergo a sequence of
transformations,
an extraction of the features (\textit{e.g.}, using Gabor
filtering~\cite{ManMa96,JaPrHoPa2000}) followed eventually by a Scale-then-Round
process~\cite{AKK20} to
accommodate representations better handled by cryptographic schemes,
\textit{i.e.}, binary or integer vectors. These templates are then protected
either through
their mere encryption, or using a Biometric Template Protection (BTP) whose goal is to address the four aforementioned security criteria. The main schemes of BTP can be categorized into three approaches:
{\it biometric cryptosystems}~(BC), {\it cancellable biometrics} (CB), and {\it keyed biometrics} (KB). In BC, a cryptographic key $K$ is bound to a biometric input $x$ through a probabilistic algorithm that takes a randomizer $r$ as additional input. A {\it helper data} HD is derived from $x$ and is stored either in a remote database or on a end-user device. This helper data serves as a protected template and should be irreversible while allowing the recovery of $x$ in the presence of another biometric input $x^\prime \approx x$. Hence, given $(r, \textrm{HD}, x^\prime)$, the key $K$ can be reproduced and used to verify the authenticity of users. Examples of BC schemes include fuzzy commitments~\cite{FuzzyCommitment}, fuzzy vaults~\cite{FuzzyVault} and fuzzy extractors~\cite{DodisRS04,apon2023nonmalleable}. In CB, a biometric input undergoes a (preferably) irreversible transformation, which can be parametrized using (public) salts or user specified secrets, and the output of the transformation is stored on the server. Verification of users are then achieved by comparing the transformed fresh biometric data with the transformed enrolled biometric data. See~\cite{JIN20042245,LUMINI2007,pillai2010sectored,yang2010dynamic} for some examples of CB schemes. In KB, secure computation techniques are used to enable the
verification in the encrypted
domain while limiting the leakage of
information~\cite{jarrous2009secure,bringer2013shade,bringer2014gshade}. The
reader
is referred to the surveys~\cite{Survey-2015,Survey-2016,patel2015cancelable}
for more details on BTP schemes.

In any biometric recognition system, the False Match Rate ($\FMR$) is an empirical assessment of the rate of false acceptations.
Nevertheless, in the context of a high-dimensional template space, unless a large number of databases are available, an empirical study does not provide reliable evidence on the range of possible configurations.
Alternatively, this work aims at providing a theoretical
treatment of
the security of binary templates as regards the occurrence of near-collisions.
To defeat a biometric system, attacks can be divided into two categories:
\begin{enumerate}
    \item Attacks targeting a single user. This is the work of Pagnin
          \textit{et~al.}~\cite{PagninDAM14} who detailed strategies to find the ball
          containing
          one secret template lying in $\mathbb{Z}^n_2$, to then mount a \textit{hill-climbing}
          attack to determine what exactly this template is (\textit{center search attack}).
    \item Untargeted attacks, in which  the goal is to impersonate
          any
          user enrolled in the secret biometric database, by the generation of a close template. This is the subject of this
          work.
          A secret biometric database may correspond to a database the attacker does not
          have access to, and for which she needs to guess a template and submit it to the
          online service to verify its admissibility (\textit{online exhaustive search
              attacks}). It may also correspond to a leaked, yet protected, database the
          attacker has access to, but the cryptographic obfuscation mechanisms for the
          verification of a protected template do not leak any information other than its
          admissibility (\textit{offline exhaustive search attacks}).
\end{enumerate}

Attack strategies considered in this work depend on the resources given to the attacker. We make the distinction between a \textit{secret transformation} and a \textit{public
    transformation}: a secret transformation relies on a secret token,
either a stored key or a memorized password, while a public transformation is
independent of any secret parameter, with the attacker having full
knowledge of it.

The amount of information measured in the literature for biometric feature vectors (or feature sets) varies according to the modality ($46$ bits for face images~\cite{adler2009towards}, $82$ bits for minutiae-based fingerprint representations~\cite{ratha2001analysis}, and $249$ bits for IrisCodes~\cite{daugman2006probing,daugman2009iris}). The present work differs in that we consider the security of templates instead of (raw) feature vectors, and under the presence of active attackers, as in~\cite{PagninDAM14}. The biometric recognition system is also assumed to be well-designed, that is, the feasibility of our attacks does not require any weakness other than letting the attacker submit a large number of guesses, exactly as in the topic of \textit{password cracking}~\cite{florencio2007strong,wang2016targeted}. Taking as assumption that templates are uniformly distributed in $\mathbb{Z}_2^n$, we provide security bounds on the size of templates, similarly to those of cryptographic hashes. Providing an exact characterization of security is to the best of our knowledge an intractable problem of combinatorics and coding-theory, due to the hard problem of determining the size of the intersection of $3$ or more Hamming balls. Although biometric data are not uniformly distributed, our formal treatment is worth consideration for the following two reasons. If the enrolled templates result from a randomized transformation (like a projection-based transformation using a per-user token, \textit{i.e.}, a salt, a secret key or a password, \textit{e.g.}, using the BioHashing scheme), they can actually be considered uniformly distributed. The second case concerns deterministic transformations. A skewed distribution of templates would certainly give some advantage to adversaries in their attack strategies, resulting in lower attack running times. However, considering attacks on uniformly distributed templates at least enables the establishment of a pessimistic lower bound on the size of templates, as well as a pessimistic upper bound on the size of the biometric database.

\paragraph*{\textit{Contributions}}

In order to provide an insight of our contributions, a summary of the most important contributions can be found in Table~\ref{Contribtable}. The presented attacks are based on exhaustive search (\textit{i.e.}, brute force attacks) and require only the minimum leakage of information, namely a bit of information about the success of impersonation. Hence, they are possible regardless of the employed BTP scheme, protocol or biometric modality.
For more details, our results fall into two distinct but complementary categories:
\begin{itemize}
    \item \textbf{Accuracy metric-based analysis}: These results focus on existing systems using their evaluated accuracy. We use well-known biometric accuracy metrics such as the $\FMR$ and $\FPIR$ to compute complexity bounds for an untargeted attack and the probability of a near-collision. The attack description provides security bounds against an outsider attacker. Computing the probability of a near-collision provides a bound on the database size to ensure a given level of security. These results are utilized to highlight the influence of the database size on such computations. These results are presented alongside an analysis of both popular literature and industry schemes, including schemes of interest for the NIST~\cite{quinn2018irex}.
    \item \textbf{Metric space-based analysis}: 
          We assume a biometric system that makes the best use of the underlying metric space in order to provide theoretical bounds on the complexity of exhaustive search attacks. We introduce the notion of weak near-collision and strong near-collision, which enable us to provide a theoretical analysis on the security strength of biometric transformation schemes. The bounds on the probability of a near-collision highlight the theoretical limits on the accuracy of biometric system. We use probabilistic modelling to present two matching attack scenarios with the associated security bounds and discuss the security of a template database. The first one, called the `Outsider Scenario', captures the case where an individual unregistered in a service attempts to impersonate a non-specific user of this service. Specifically, we consider the possibility of an attacker sequentially adapting her strategy. The second scenario, termed the `Insider Scenario', encapsulates cases where some or all users of the service are potential adversaries attempting to impersonate one another. The bounds on the complexity of the untargeted attacks provide the maximum achievable security.
          Finally, we make recommendations concerning the security parameters  during the fine-tuning of a recognition system.
\end{itemize}


\paragraph{\textit{Scope of our results}}
Our results on the complexity of untargeted attacks and the probability of near-collision occurrences apply to many BTP schemes. To the best of our knowledge, many BTP schemes of the three categories (CB, BC and KB) are vulnerable to \textit{offline} attacks regardless of the considered modality.
For instance, among the BC schemes, we can identify fuzzy commitments~\cite{FuzzyCommitment}, fuzzy vaults~\cite{FuzzyVault} and fuzzy extractors~\cite{DodisRS04,CanettiFPRS16,apon2023nonmalleable}, to name but a few. Concerning the attacks, they can be performed either \textit{online} or \textit{offline}, and the derived complexity results apply in both cases.
Offline attacks are made possible when the protected biometric database is leaked, as the attacker exploits some (even minimal) information that allows her to test a guess. For the near collision, it highlight a theoretical limit on the performance of biometric recognition algorithms.

\paragraph*{\textit{Outline}}

Section~\ref{prel} introduces notations, background material and definitions related to near-collisions and biometric transformation schemes.
Section~\ref{section:FMR} provides security bounds for untargeted attacks and near-collisions, using the $\FMR$ metric along with a discussion on the security of relevant examples from the literature.
Section~\ref{utg_attacks} is devoted to a refinement of template security considering untargeted attacks in the metric space setup. It is described how near-collisions can be used to determine security bounds depending on the template space dimension, the decision threshold and the number of registered users in the service.
In Section~\ref{template_sec}, numerical evaluations are provided about the security of biometric databases. Especially, some recommendations are given to parameterize correctly biometric transformation schemes.
Section~\ref{conclusion} concludes the paper.

\section{Preliminaries}
\label{prel}
\begin{table*}[h]
    \begin{tabular}
        {N N N N N N}
        \toprule
        \multicolumn{6}{c}{\textbf{Accuracy metric-based complexities}}                                                                                                                                                                                          \\ \cmidrule{1-6} 
        \multicolumn{3}{N}{Outsider scenario}                                       & \multicolumn{3}{N}{Database limit size}                                                                                                                                    \\
        \multicolumn{3}{N}{Theorem~\ref{db_size_outsider_independent_events_FMR}}   & \multicolumn{3}{N}{Theorem~\ref{FMR_and_NC_gives_N}}                                                                                                                       \\\cmidrule{1-6} 
        \multicolumn{6}{c}{\textbf{Metric space-based  complexities}}                                                                                                                                                                                            \\ \cmidrule{1-6} 

        \multicolumn{3}{N}{Outsider scenario}                                       & \multicolumn{3}{N}{Insider scenario}                                                                                                                                       \\ 

        \multicolumn{3}{N}{Theorem~\ref{db_size_outsider_independent_events}}       & \multicolumn{3}{N}{Theorem~\ref{thm:insider_attack:m_bounds}}                                                                                                              \\

        \multicolumn{3}{N}{Corollary~\ref{db_size_outsider}}                        & \multicolumn{3}{N}{Corollary~\ref{thm:insider_attack:m_bounds:l_attaquant}}                                                                                                \\ \midrule

        \multicolumn{6}{c}{\textbf{Metric space-based naive and adaptive attackers equivalence}}                                                                                                                                                                 \\ \cmidrule{1-6} 
        \multicolumn{2}{c}{Median number of trials}                                 & \multicolumn{2}{N}{Cumulative distribution function}                                   & \multicolumn{2}{c}{Probability mass function}                                     \\
        \multicolumn{2}{c}{ Theorem~\ref{cor:naive_eq_k_adap:median_and_expect}}    & \multicolumn{2}{c}{Proposition~\ref{prop:outsider_attack:negligible_difference_range}} & \multicolumn{2}{c}{ Proposition~\ref{prop:outsider_attack:negligible_difference}} \\  \midrule

        \multicolumn{6}{c}{\textbf{Metric space-based exact probability bounds}}                                                                                                                                                                                 \\ \cmidrule{1-6} 
        \multicolumn{3}{N}{Weak near-collision}                                     & \multicolumn{3}{N}{Master template occurrence}                                                                                                                             \\
        \multicolumn{3}{N}{Proposition~\ref{prop:bound:nearcol}}                    & \multicolumn{3}{N}{Theorem~\ref{thm:proba:NMT}}                                                                                                                            \\
        \multicolumn{3}{N}{Corollary~\ref{thm:insider_attack:m_bounds:l_attaquant}} & \multicolumn{3}{N}{Corollary~\ref{cor:proba:kMT}}
        \\\bottomrule
    \end{tabular}
    \caption{Summary of the main contributions}
    \label{Contribtable}
\end{table*}

\subsection{Notations}

Let $q$ be an integer in $\mathbb{N}_{>2}$. The set $\mathbb{Z}_q^n$ corresponds to the $n$-dimensional vector space over $\mathbb{Z}_q = \lbrace 0,\dots,q-1 \rbrace$. In the following, the binary case ($q=2$) is always explicitly written as $\mathbb{Z}_2$, so that $\mathbb{Z}_2^n$ denotes the set of binary vectors of length $n$.
Given the Hamming distance $d_\mathcal{H}~: \mathbb{Z}_2^n \times \mathbb{Z}_2^n \rightarrow \mathbb{N}$, a vector $t \in \mathbb{Z}_2^n$ and a positive number $\varepsilon \in \mathbb{N}$, the Hamming ball of center~$t$  and radius~$\varepsilon$ is defined by $B_{\varepsilon}(t)=\{y \in \mathbb{Z}_2^n~:~ d_{\mathcal{H}}(t,y) \leq \varepsilon \}$. We denote by $\mathcal{B}(\mathbb{Z}_2^n,d_\mathcal{H})$ a biometric database whose enrolled templates lie in the template metric space $(\mathbb{Z}_2^n,d_\mathcal{H})$, equipped with the Hamming distance $d_\mathcal{H}(.,.)$ and the threshold $\varepsilon$ for the comparison of two templates. We also define $V_\varepsilon$ as the measure of an $\varepsilon$-ball in $\mathbb{Z}^n_2$, \textit{i.e.} $|B_{\varepsilon}|/2^n$ where $|B_{\varepsilon}|$ is the cardinal of an $\varepsilon$-ball.

\subsection{Biometric Template Protection Schemes and Operating Recognition Modes}

The templates in $\mathcal{B}$ result from a chain of treatments, an extraction of the features (\textit{e.g.}, using Gabor filtering~\cite{ManMa96,JaPrHoPa2000}) followed eventually by a Scale-then-Round process~\cite{AKK20} to accommodate better handled representations, usually binary vectors. These feature vectors can also be transformed into a binary form using a tokenized projection-based protection scheme, where tokens may be secret or public. Thus, a significant part of biometric systems are based on binary templates which motivates the present work. Depending on the confidence level in the security properties (irreversibility and unlinkability) offered by this composition of transformations, these binary templates may be further protected using a cryptographic scheme, or using another BTP scheme that rely on a standard hardness assumption.

For the sake of simplicity, argument and illustration in the following sections, only the definition of a CB scheme is recalled in detail.

\begin{definition}[Cancellable Biometric transformation scheme]
    Let $\mathcal{K}$ be the token (seed) space, representing the set of tokens to be assigned to users. A CB scheme is a pair of deterministic polynomial time algorithms $\Xi:=(\mathcal{T}, \mathcal{V})$.
    The first element is $\mathcal{T}$, the transformation of the system, that takes a feature vector $x$, and a token $s$ as input, and returns a biometric template $t=\mathcal{T}(s,x) \in \mathbb{Z}_2^n$.
    The second element is $\mathcal{V}$, the verifier of the system, that takes two biometric templates $t$ = $\mathcal{T}(s,x)$, ${t^\prime}=\mathcal{T}({s^\prime},{x^\prime})$, and a threshold $\varepsilon$ as input; and returns
    $True$ if $d_{\mathcal{H}}(t, {t^\prime}) \leq \varepsilon$, and
    returns $False$ if $d_{\mathcal{H}}(t, {t^\prime}) > \varepsilon$.
\end{definition}

The operations of a CB scheme do not exactly capture the functionalities of other BTP schemes, as KB schemes or BC schemes. In KB schemes, the underlying algorithms require cryptographic keys as additional inputs. As regards some KB schemes based on homomorphic encryption (or related primitives like function-hiding inner-product functional encryption (fh-IPFE~\cite{barbosa2019efficient})), $\mathcal{T}$ requires a public key to encrypt a template, and $\mathcal{V}$ requires a private key to evaluate the comparison function.
\medskip

After the initial enrollment of a user consisting of the storage
of a reference template in the database $\mathcal{B}$, the biometric recognition system
may operate either in the \textit{verification/authentication mode}
or in the \textit{identification mode}:
\begin{itemize}
    \item In the authentication mode, the person’s
          identity is validated by comparing a freshly acquired template with her
          own reference template stored in $\mathcal{B}$.
    \item In the identification mode, an individual is recognized by searching
          the templates of all the users in $\mathcal{B}$ for a match.
          Therefore, the system conducts a
          one-to-many comparison to find an individual’s identity, or fails
          if the subject is not enrolled in $\mathcal{B}$,
          without the subject having to claim her identity.
\end{itemize}

We consider both the operating modes and a two party setting (\textit{i.e.}, one client $C$ and one authentication or
identification server~$S$), but our analyses also apply when more than two
parties are involved in the biometric recognition process~\cite{BringerCIPTZ07,barbosa2008secure,stoianov2009security}.
When the biometric transformation scheme does not provide formal guarantees on the privacy of users,
the biometric templates should not be sent in the clear over the network. This implies that often the
matching operation is performed in the encrypted domain. Biometric recognition protocols employ
secure multi-party computation techniques to preserve the privacy of the users, for instance with the use of
homomorphic encryption~\cite{paillier1999public}, garbled circuits~\cite{yao1986generate} or oblivious transfer~\cite{rabin2005exchange}.
No matter the privacy-preserving techniques employed in the protocols, they do not mitigate our attacks since they
rely solely on the binary outcome of the matching, \textit{i.e.}, the acceptation or rejection result of the server.

\subsection{Security Metrics in Biometrics}
\label{Metrics_definition}
To assess the security of a biometric system, different metrics are used based on the operation mode (recognition or identification).
In the context of authentication systems (specifically, in a 1:1 system), the predominant metric utilized is the False Match Rate ($\FMR$). This rate serves as an empirical estimation denoting the likelihood of a biometric sample being incorrectly recognized by the matcher.
In other words, this is an estimation of the probability that the matcher incorrectly decides that a newly collected template matches the stored reference.
The False Non-Match Rate ($\FNMR$) 
gives an estimation of the probability that a genuine sample is incorrectly rejected by the matcher.

The confusion between the False Match Rate ($\FMR$) and the False Acceptance Rate ($\FAR$) often arises due to their subtle distinction. The $\FAR$ operates at the system-wide scale, encompassing more than the matching component. Specifically, it represents the probability of a biometric sample being falsely recognized by the entire system. The $\FAR$ considers the collective performance of all security layers within a biometric system, such as liveness detection, thereby providing a comprehensive assessment of system integrity.
Let $\FTA$ denote the Failure To Acquire rate \textit{i.e.}, the probability that the system fails to produce a sample of sufficient quality. We have the following equality
$$\FAR = \FMR \times (1-\FTA).$$
The misconception between the False Non-Match Rate ($\FNMR$) and the False Reject Rate ($\FRR$) persists for analogous reasons as discussed above.
Those two notions can be linked by
$$\FRR = \FTA + \FNMR\times(1-\FTA).$$

The concepts mentioned earlier rely on a threshold selected to minimize either the False Non-Match Rate ($\FNMR$) or the False Match Rate ($\FMR$). Typically, this threshold is determined at the Equal Error Rate ($\EER$) where the $\FNMR$ equals the $\FMR$.

In the context of identification (specifically, in a 1:$N$ system), a widely used metric is the False Positive Identification Rate ($\FPIR$). This metric quantifies the errors rate when the system misidentifies an impostor as a user. The $\FPIR$ shares similarities with the False Match Rate ($\FMR$) as they both evaluate recognition errors within biometric systems.
Similarly to the False Non-Match Rate ($\FNMR$) in verification (1:1 system), the False Non-Identification Rate ($\FNIR$) assesses the likelihood of genuine users being incorrectly rejected or failing to be identified by the system in the identification mode (1:$N$). The $\FNIR$ is often correlated with the $\FRR$ as they both depict the system's accuracy.


\subsection{Adversarial Model}
\label{Attacker_setup}
As introduced by Simoens~\textit{et~al.}~\cite{simoens2012framework}, attacks on biometric authentication systems can be divided into $4$ categories, of which the two strongest are our main focus:
\begin{enumerate}
    \item \textit{Biometric reference recovery:} the attacker tries to recover the stored enrolled template.
    \item \textit{Biometric sample recovery:} the attacker tries to recover (or generate) a fresh biometric template acceptable by the biometric authentication system.
\end{enumerate}

\paragraph{\textit{Resources of the Attacker}} These categories can be extended to identification systems.
The considered attacker exploits the fourth attack
point in the model of a generic biometrics-based system, as presented by Ratha~\textit{et~al.}~\cite{Rathasheme}.
More specifically, the attacker can provide any information to the matcher but does not have access to the database (in cleartext) under any circumstances.
She also has access to the matcher's responses to her queries.
If some extra information is required, such as login for the authentication mode, it is assumed that
the attacker has knowledge of it.

\paragraph{\textit{Offline and Online Attacks}}
In an online attack, the online service is used by the attacker to test her queries.
Offline attacks can be exploited by attackers after the leakage of elements from the remote database, even if they are leaked in an altered or transformed form that preserves their secrecy~\cite{tams2016unlinkable}. In such scenario, the attacker can evaluate a function that is mostly zero if the function is well designed (\textit{i.e.} obfuscated) for privacy purposes. We consider in this paper functions that leak the minimum leakage of information, for instance the binary outcome of a verification function, based on a distance or a checksum/hash function. By way of examples, in the case of a BC scheme, a checksum is often used to verify if the correct key is unlocked  or reproduced (\textit{e.g.}, fuzzy extractors~\cite{DodisRS04,CanettiFPRS16,apon2023nonmalleable} and fuzzy commitments~\cite{FuzzyCommitment}).

\paragraph{\textit{From Untargeted Sample Recovery to Untargeted Reference Recovery}} Sample recovery and reference recovery attacks in the authentication mode have been introduced by Pagnin \textit{et al.}~\cite{PagninDAM14}
They rely on a single bit of information leakage and are applicable in the targeted scenario.
The present work fills a gap by putting emphasis on \textit{biometric sample recovery}, still using a single bit of leakage, but applicable in the untargeted scenario and in both the operating modes.
Our \textit{untargeted sample recovery attacks} can be translated into \textit{untargeted reference recovery attacks} almost for free:
\begin{itemize}
    \item Concerning a biometric system built on top of CB scheme or a KB scheme, we merely apply the center search attack of~\cite{PagninDAM14}. Once a sample template of $\mathbb{Z}_2^n$ of any user is found, the reference template is recovered in at most $n+2\varepsilon$ calls to the verifier function. (It is worth noting that if the attacker has access to more than the binary outcome of the verifier function, \textit{e.g.}, the result of the distance function below the threshold, then a hill climbing attack can be launched, hence reducing the number of calls to $2\varepsilon$.)
    \item In the case of a biometric system built on top of a BC scheme (\textit{e.g.}, a fuzzy commitment or a fuzzy vault) and under the leakage of the protected database, the reference template is derived from the sample template by the mere application of the recovery function.
\end{itemize}

\subsection{Relations with Near-collisions}

The database $\mathcal{B}$ contains $N$ templates which are
vectors distributed in $\mathbb{Z}_2^n$. This distribution is considered uniform
if templates result from a salted or a secret-based transformation.
If the chain of treatments applied to the feature vectors is deterministic, the
templates can actually be considered as non-uniformly distributed in
$\mathbb{Z}_2^n$.
Some terminologies and definitions are introduced for the
sake of our analyses.

\begin{definition}[Strong $\varepsilon$-near-collision for $t$]
    \label{def:strong:Ncol}
    For a secret template $t$ of $\mathcal{B}$,
    a strong near-collision occurs if there is another template $a \in \mathcal{B}$ such that $d_\mathcal{H}(t,a) \leq \varepsilon$.
\end{definition}
Notice that for a given secret template $t$ of a database $\mathcal{B}$, the probability of $C^N(t)$ the event "At least one of the $N-1$ other users of $\mathcal{B}$ matches the given template $t$" is given by
\begin{equation*}
    \mathbb{P}\left(C^{N} (t)\right) = 1-\left(1-V_\varepsilon\right)^{N-1}.
\end{equation*}
In other words, in the case of this definition the occurrence of a strong $\varepsilon$-near-collision increases geometrically.
In the following,  different near-collision events (Definitions~\ref{NCol} and \ref{MNCol}) are considered for which it is not possible to smoothly derive probability and complexity results.
Analyses providing our results are detailed in Section~\ref{utg_attacks} and Section~\ref{mlty_near_collisions}.

\begin{definition}[Weak $\varepsilon$-near-collision]
    \label{NCol}
    For $\mathcal{B}$ a biometric database, a weak $\varepsilon$-near-collision is occurring in $\mathcal{B}$ if there exists two templates
    $a,b \in \mathcal{B}$ such that $d_\mathcal{H}(a,b) \leq \varepsilon$.
\end{definition}

In other words, a weak near-collision occurs if there exists a pair of
templates $a$ and $b$ in the secret database such that $a$ (resp. $b$) is inside
the ball of center $b$ (resp. $a$).
Note that if $a$ and $b$ weak near-collides, then we have two strong near-collisions, for $a$, and for $b$.
This definition can be generalized to the
case of multi-near-collisions.

\begin{definition}[Weak $(m,\varepsilon)$-near-collision]
    \label{MNCol}
    For $\mathcal{B}$ a biometric database, an \textit{$(m,\varepsilon)$-near-collision} with $m \geq 2$ is occurring if there exists an $m$-tuple $a_1,\dots,a_m \in \mathcal{B}$
    such that for all
    $i$, $j \in \lbrace 1,\dots,m \rbrace,$ $d_\mathcal{H}(a_i,a_j) \leq \varepsilon$.
\end{definition}

For privacy purposes, a security criteria stated in ISO/IEC
24745~\cite{ISO24745} and ISO/IEC 30136~\cite{ISO30136} is the irreversibility
of templates. The notions introduced above can be used to characterize the
reversibility of $\Xi.\mathcal{T}$, and this motivates the following
definitions.

\begin{definition}[Strong $\varepsilon$-nearby-template preimage]
    For a template $t$ of $\mathbb{Z}^n_2$,
    a strong nearby-template preimage is a pair of token $s$ and feature vector $x$ such that $\mathcal{T}(s,x)$
    is a strong $\varepsilon$-near-collision for $t$.
\end{definition}

\begin{definition}[Weak $\varepsilon$-nearby-template preimages]
    There exists two pairs of a feature vector and a token $(s_1,x_1)$ and $(s_2,x_2)$ with $(s_1,x_1) \not = (s_2,x_2)$
    such that $t_1 = \mathcal{T}(s_1,x_1)$ and $t_2 = \mathcal{T}(s_2,x_2)$ correspond to a weak $\varepsilon$-near-collision.
\end{definition}

\section{Bounds with Respect to the $\FMR$, $\FPIR$ and Near-Collisions}
\label{section:FMR}

This section deals with the \textit{public transformation} case, where the attacker has access to the transformation (\textit{i.e.}, the transformation does not rely on a second factor) in the biometric sample recovery scenario. The case where the attacker tries to generate a template which is accepted by the matcher is studied first. Then, the attacker tries to be accepted by the system. Finally, we show how these attack complexities apply to existing systems.

\subsection{$\FMR$-based Attack Complexities}
Suppose that an attacker as defined in Section~\ref{Attacker_setup} has access to the matcher. The attacker can then produce any template and call the matcher to compare it to all the templates in the database. Each of the $N$ users $i$ in the database has its own $\FMR$ denoted $\FMR_i$ (see Section~\ref{Metrics_definition}). Then, the probability that the attacker matches the template of the user $i$ is $\FMR_i$. It follows that the probability that she does not matches any template is $\prod_{i=1}^{N} (1-\FMR_i)$. Then the probability that she matches at least one template is $1-\prod_{i=1}^{N} (1-\FMR_i)$. The following result can be stated.
\begin{lemma}\label{db_size_outsider_independent_events_FMR}
    Let $N$, $n$ and $(\FMR_i)_{i=1}^N$ be fixed parameters as above. Then, the median number $m_{out}$ of trials for an outsider attacker to successfully impersonate a user is
    $$\Omega\left( 2^{-\log_2 \left(\sum_{i=1}^{N}\left(\FMR_i(1+\FMR_i)\right)\right)}\right)
        \text{ and }
        O\left( 2^{-\log_2 \left(\sum_{i=1}^{N}\FMR_i\right)}\right)$$
\end{lemma}

\begin{proof}
    Recall that the median of a geometric law is $m_{out} \approx \frac{-1}{\log_2(1-p)}$.
    The probability of success with a single trial is given by
    $$p=1-\prod_{i=1}^{N} (1-\FMR_i)$$ as explained above.
    Since
    \begin{equation*}
        -x-x^2 \leq \log(1-x) \leq -x \text{, for } 0 \leq x \leq \frac{1}{2},
    \end{equation*}
    and,
    $$\log_2(1-p) = \sum_{i=1}^{N}\log_2 (1-\FMR_i),$$
    then,
    \begin{eqnarray*}
        \frac{-\sum_{i=1}^{N} \FMR_i(1+\FMR_i)}{\ln 2} \leq & \log_2(1-p) & \leq \frac{-\sum_{i=1}^{N} \FMR_i}{\ln 2}\\
        \ln (2) \left(\sum\limits_{i=1}^{N} \FMR_i(1+\FMR_i)\right)^{-1} \leq & m_{out} & \leq \ln (2) \left(\sum\limits_{i=1}^{N} \FMR_i\right)^{-1}.\\
    \end{eqnarray*}
    Then, we have
    \begin{eqnarray*}
        2^{-\log_2 \left(\sum_{i=1}^{N}\left(\FMR_i(1+\FMR_i)\right)\right)+\log_2(\ln2)}
        & \leq & m_{out} \\
        m_{out} & \leq &
        2^{-\log_2 \left(\sum_{i=1}^{N}\FMR_i\right)+\log_2(\ln 2)}\\
    \end{eqnarray*}
    and the number of rounds follows.
\end{proof}



Using the result of Lemma~\ref{db_size_outsider_independent_events_FMR}, and assuming that the $\FMR$ is the same for every user, the following theorem can be stated.
\begin{theorem}
    Let $N$, $n$ and $\FMR$ be fixed parameters as above. Then, the median number $m_{out}$ of trials for an outsider attacker to successfully impersonate a user is
    $$\Omega\left( 2^{- \log_2{(\FMR)} -\log_2{(1+\FMR)} -\log_2{(N)}}\right)
        \text{ and }
        O\left( 2^{-\log_2 (\FMR) -\log_2 (N)}\right)$$
\end{theorem}

\begin{proof}
    Using Lemma~\ref{db_size_outsider_independent_events_FMR} and the fact that for all $i,j$, $\FMR_i = \FMR_j$, the result follows.
\end{proof}

Another result can be derived using the $\FTA$ to obtain the number of rounds by considering the $\FAR$.
\begin{proposition}
    Let $N$, $n$, $(\FMR_i)_{i=1}^N$ and $(\FTA)_{i=1}^N$ be fixed parameters as above. Then, the median number $m_{out}$ of trials for an outsider attacker to successfully impersonate a user is
    $$\Omega\left( 2^{-\log_2 \left(\sum_{i=1}^{N}\left(\frac{\FAR_i}{1-\FTA_i}\left(1+\frac{\FAR_i}{1-\FTA_i}\right)\right)\right)}\right)
        \text{ and }
        O\left( 2^{-\log_2 \left(\sum_{i=1}^{N}\frac{\FAR_i}{1-\FTA_i}\right)}\right)$$
\end{proposition}

\begin{proof}
    Using Lemma~\ref{db_size_outsider_independent_events_FMR} and the fact that $\FAR = \FMR\times(1-\FTA)$, the result follows.
\end{proof}

\subsection{$\FPIR$-based Attack Complexities}
Suppose that an attacker as defined in Section~\ref{Attacker_setup} has access to the matcher in the identification mode. The attacker can then produce any template and compare it to all the templates in the databases. The probability that the attacker match a user at each round is equal to the $\FPIR$. Hence, the following result can be stated.
\begin{theorem}
    \label{db_size_outsider_independent_events_FPIR}
    Let $N$, $n$ and $(\FPIR)$ be fixed parameters as above. Then, the median number $m_{out}$ of trials for an outsider attacker to successfully impersonate a user is
    $$\Omega\left( 2^{-\log_2(\FPIR)-\log_2(1+\FPIR)}\right)
        \text{ and }
        O\left( 2^{-\log_2(\FPIR)}\right)$$
\end{theorem}

\begin{proof}
    The proof is the same as the proof of Lemma~\ref{db_size_outsider_independent_events_FMR} with $p=\FPIR$.
\end{proof}

Note that the link between the $\FMR$ and the $\FPIR$ is not trivial and the $\FMR$ should not be used to infer the $\FPIR$ because of the potential near-collisions (see Definition~\ref{NCol} and Definition~\ref{MNCol}).

\subsection{Probability of a Near-Collision with Respect to Both the $\FMR$ and the Database Size}

A near-collision occurs when two or more users of a given database can potentially authenticate or be identified in place of each other leading to accuracy and security problems.
\begin{proposition}
    \label{FMR_and_NC}
    Let $N$, $n$ and $(\FMR)$ be fixed parameters as above. Then, the probability that there is a near-collision is $$1-(1-\FMR)^{N(N-1)/2}.$$
\end{proposition}

\begin{proof}
    Considering a uniformed $\FMR$ and using the birthday problem, the probability of a near collision can be found. For a database of size $N$, the number of distinct pairings is $N(N-1)/2$. According to the $\FMR$ definition, the probability of a pair to result in a False Match is $\FMR$. Hence, the probability that there is no near-collision is $(1-\FMR)^{N(N-1)/2}$ and the result follows.
\end{proof}

In order to achieve a given level of security and accuracy with respect to the occurrence of near-collision, the number of clients in a database must be bounded. Considering $\lambda\in\mathbb{N}^\times$ a secure parameter such that the probability of a near collision is smaller than $1/\lambda$, the maximum size for a database is stated with respect to its $\FMR$.

\begin{theorem}
    \label{FMR_and_NC_gives_N}
    Let $N$, $n$, $\lambda$ and $(\FMR)$ be fixed parameters as above. If
    $$N \leq \frac{1}{2}\left(1+\sqrt{1+\left(8\log_2(1-1/\lambda)\right)/\left(\log_2(1-\FMR)\right)}\right)$$
    then the probability that there is a near-collision is smaller than $1/\lambda$.
\end{theorem}

\begin{proof}
    Using the previous theorem, we seek $N$ such that
    \begin{eqnarray*}
        1-(1-\FMR)^{N(N-1)/2} & \leq & 1/\lambda\\
        N^2-N+\left(2\log_2(1-1/\lambda)\right)/\left(\log_2(1-\FMR)\right) & \leq & 0\\
    \end{eqnarray*}
    The study of the last function yields the result.
\end{proof}

Corollary~\ref{Cor:assymptotic_N_bound} gives an asymptotic estimation for the bound over $N$.

\begin{corollary}
    \label{Cor:assymptotic_N_bound}
    Let $N$, $n$, $\lambda$ and $(\FMR)$ be fixed parameters as above. If
    $$N \approx \sqrt{2/\left(\lambda\FMR\right)}$$
    then the probability that there is a near-collision is smaller than $1/\lambda$.
\end{corollary}

\begin{proof}
    Using the previous theorem, we have $N$ such that
    \begin{eqnarray*}
        N & \leq & \frac{1}{2}\left(1+\sqrt{1+\left(8\log_2(1-1/\lambda)\right)/\left(\log_2(1-\FMR)\right)}\right)\\
    \end{eqnarray*}
    Moreover,
    \begin{eqnarray*}
        \frac{1}{2}\left(1+\sqrt{1+\left(8\log_2(1-1/\lambda)\right)/\left(\log_2(1-\FMR)\right)}\right) & \approx & 
        \sqrt{2/\left(\lambda\FMR\right)}\\
    \end{eqnarray*}
    and the result follows.
\end{proof}



\subsection{Numerical Evaluations on Real-world Systems}

In order to support the two above theorems, Table~\ref{Table_FMR_result} shows the bounds on the number of rounds for an attacker on several real-life systems along with the probability of a near collision and the recommendations for the choice of $N$ with $\lambda=128$.
The evaluation metrics rely on industry giants like Google and Apple, state-of-the-art via reviews, standards from NIST, and the FVC ongoing platform to identify best algorithms. We systematically excluded algorithms with an $\FMR$ of $0$ to avoid numerical approximation biases. We selected algorithms with the lowest {\tt EER} (Equal Error Rate) for several modalities in order to draw a representative picture of the current state of biometrics. When reasonable (\textit{i.e.} when the $\FMR$ is smaller than $1\%$), we considered the $\FMR$ scores at the {\tt ZeroFNMR} (the lowest $\FMR$ for ${\tt FNMR}=0\%$) to analyze the most practical algorithms. When the $\FMR$ is too high to consider the {\tt ZeroFNMR}, we take the {\tt EER}. If we do not have access to any of these information, we consider the given $\FMR$.

According to the NIST~\cite{quinn2018irex}, the most accurate one-to-one iris matcher in $2018$ yields an $\FMR$ of $10^{-5}$ ($1$ in $100,000$) for  an enrolled population of $160$ thousand people. Concerning the face and the fingerprint modalities, Apple claims respectively an $\FMR$ of $10^{-6}$ ($1$ in $1,000,000$) for the face ID~\cite{FaceID} and an $\FMR$ of $5\times 10^{-5}$ ($1$ in $500,000$) for the touch ID~\cite{TouchID}. For android, Google claims that their systems on several modalities (Iris, Fingerprint, Face and Voice) ensure an $\FMR$ lower than $2\times 10^{-5}$ ($1$ in $200,000$)~\cite{AndroideID}. The state-of-the-art $\FMR$ for Vascular Biometric Recognition (VBR) is $10^{-4}$ ($1$ in $10,000$)~\cite{quinn2018irex} on a database of $100$ individual. According to Sandhya and Prasad~\cite{sandhya2017biometric}, the Adaptive selection of error correction (ASEC) on Online signature for $30$ subject leads to a $\FMR$ of $4\%$. The Enhanced BioHash NXOR mask~\cite{sandhya2017biometric} applied to face recognition evaluated with $294$ users leads to a $\FMR$ of $0.11\%$ (approximately $1$ in $1,000$). Minutiae descriptors~\cite{sandhya2017biometric} for fingerprint on $100$ users yields an $\FMR$ of $10^{-4}$ ($1$ in $10,000$). According to FVC-onGoing~\cite{dorizzi2009fingerprint,FVC}, for the {\tt ZeroFNMR}, HXKJ for fingerprint gets a $\FMR$ of $0.005\%$ ($5$ in $100,000$) evaluated on the FV-STD-$1.0$ database. The MM\_PV~\cite{FVC} algorithm for palm vein at the {\tt EER} gives an $\FMR$ of $0.328\%$ (approximately $3$ in $1,000$) evaluated on the PV-FULL-$1.0$ database. The Biotope~\cite{FVC} algorithm for Secure Template Fingerprint (STF) at the {\tt EER} gives a $\FMR$ of $1.541\%$ (approximately $15$ in $1,000$) evaluated on the STFV-STD-$1.0$ database.

We estimate the plausible number of clients for the databases whose the actual number is not available.
Since the security bounds rely on this estimation, the actual security remains uncertain. In the following, a description is provided on how the estimates are formulated.
Concerning Apple, a smartphone typically permits only $5$ to $10$ distinct enrolled users, while the Android system allows $2$ to $4$.
In the case of the FVC~\cite{FVC} databases used for benchmarking, where the number of genuine attempts is specified, based on the factorization of the genuine attempts number, we assumed that each individual possesses approximately $10$ samples, $2772$ users for FV-STD-$1.0$ and STFV-STD-$1.0$ are found. Those results stem from the division of the number of genuine attempts by the number of samples.
Using the same methodology, we assume that there are only $10$ samples per user, yielding a total of $280$ client for PV-FULL-$1.0$.

\begin{table*}[]
    \resizebox*{\textwidth}{!}{%
        \begin{tabular}{@{}lrrcccccc@{}}
            \toprule
            \multicolumn{1}{c}{\multirow{3}{*}{\FMR}} & \multicolumn{1}{c}{\multirow{3}{*}{System}} & \multirow{3}{*}{Modality} & DB        & DB        & Number of       & Number of       & Probability of & Maximum N       \\
            \multicolumn{1}{c}{}                      & \multicolumn{1}{c}{}                        &                           &           & estimated & rounds (lower   & rounds (upper   & Near collision & for \FMR        \\
            \multicolumn{1}{c}{}                      & \multicolumn{1}{c}{}                        &                           & size      & size      & bound $\log_2$) & bound $\log_2$) & ($\%$)         & $(\lambda=100)$ \\
            \midrule
            $1$ in $5.0\times 10^5$                   & Apple Touch ID~\cite{TouchID}               & Fingerprint               & NA        & $7$       & $16$            & $16$            & $0.004$        & $100$           \\
            $1$ in $1.0\times 10^6$                   & Apple Face ID~\cite{FaceID}                 & Face                      & NA        & $8$       & $16$            & $16$            & $0.002$        & $142$           \\
            $1$ in $2.0\times 10^5$                   & Google standard~\cite{AndroideID}           & All                       & NA        & $3$       & $16$            & $16$            & $0.001$        & $63$            \\
            $1$ in $1.0\times 10^5$                   & $NEC5$~\cite{quinn2018irex}                 & Iris                      & $160,000$ & $-$       & $0$             & $0$             & $-$            & $45$            \\
            $1$ in $1.0\times 10^4$                   & Nikisins~\cite{nikisins2018fast}            & VBR                       & $100$     & $-$       & $6.64$          & $6.64$          & $39.044$       & $14$            \\
            $4$ in $1.0\times 10^2$                   & ASEC~\cite{sandhya2017biometric}            & Online signature          & $30$      & $-$       & $0$             & $0$             & $99.999$       & $1$             \\
            $11$ in $1.0\times 10^4$                  & NXOR~\cite{sandhya2017biometric}            & Face                      & $294$     & $-$       & $1.63$          & $1.63$          & $100$          & $4$             \\
            $5$ in $1.0\times 10^5$                   & HXKJ~\cite{FVC}                             & Fingerprint               & NA        & $2772$    & $2.85$          & $2.85$          & $100$          & $20$            \\
            $3$ in $1.0\times 10^4$                   & MM\_PV~\cite{FVC}                           & Palm vein                 & NA        & $2772$    & $0.26$          & $0.27$          & $100$          & $8$             \\
            $15$ in $1.0\times 10^4$                  & Biotope~\cite{FVC}                          & STF                       & NA        & $280$     & $1.24$          & $1.25$          & $100$          & $4$             \\ \bottomrule
        \end{tabular}
    }
    \caption{Number of trials for the attackers on real examples.}
    \label{Table_FMR_result}
\end{table*}

\section{Theoretical Matching Attacks}
\label{utg_attacks}

Below are presented some untargeted attacks to find near-collisions with
hidden templates of a secret biometric database.
We examine two attack scenarios, estimating the bounds for their respective run-time complexities.
In the first scenario, an outsider attacker submits guesses to the system until one of them is accepted.
In the insider scenario, some or all of the genuine users launch the same attack and try to impersonate others.
They apply irrespective
of the operation mode (identification or verification).
However, unlike identification, authentication requires the set of identifiers (logins).
In this case, the attacker needs to test
a guessed template for each of the identifiers, hence adding a factor of $N$ in
the estimated bounds.

\subsection{Naive and Adaptive Attack Models}
\label{utg_attacks:out}

The attacker $\mathcal{A}$ is an outsider of the system, \textit{i.e.}, she is not 
enrolled in the system, and she seeks to perform an untargeted attack by
impersonating any of the $N$ users in the database. For the considered attack, the database
is not leaked in cleartext, so the templates remain secret.
It is assumed that $\mathcal{A}$ generates her own database, a number of templates $t_1^a$, $t_2^a, \dots, t_k^a$ via the function $\mathcal{T}$ until one of them is accepted by the system. To do so, she generates inputs randomly and applies $\mathcal{T}$ on these inputs. In the following, we calculate the probability of a strong near-collision for the generated template as well as the number of trials of a naive attacker.

\subsubsection{Naive Attacks}

We denote by $E_1^N$ the event "the template of the outsider matches with at least one of the $N$ templates of the database".
When no constraints are imposed on enrollment templates, $E_1^N$ is seen as a union of independent events $E_1^1$.
Consider that the attacker repeats this attack with each new generation of a template until achieving success.
According to the geometric distribution, $m_\text{out}$ the necessary number of templates to generate so that
a success occurs with more than a chance of $50\%$, corresponds to the median number of trials to succeed, \textit{i.e.}, about $-\log(2)/\log(1-p)$ where $p$ is the probability of success with a trial.

\begin{theorem}\label{db_size_outsider_independent_events}
    Let $N$, $n$ and $\varepsilon$ be fixed parameters with $\varepsilon/n \leq 1/2$ and $N < 2^{n(1-h(\varepsilon/n))-1}$. The median number $m_\text{out}$ of trials for the attacker to successfully impersonate a user is
    $$\Omega\left(2^{n(1-h(\varepsilon/n))-\log_2 N}\right)
        \text{ and } O\left(2^{n(1-h(\varepsilon/n))  + \frac{1}{2} \log_2 \varepsilon \left(1- \frac{\varepsilon}{n} \right) - \log_2 N }\right)$$
    where $h(\cdot)$ is the binary entropy function.
\end{theorem}

\begin{proof}
    The probability of success with a single trial is given by
    $$p=1-(1-V_{\varepsilon})^N.$$
    Since $-x-x^2 \leq \log(1-x) \leq -x$ for $0\leq x \leq \frac{1}{2}$ and if $N < 2^{n(1-h(\varepsilon/n))-1}$,
    then
    \begin{equation*}
        - N (V_\varepsilon - V_\varepsilon^2) \leq N \log(1-V_\varepsilon) \leq - NV_\varepsilon
    \end{equation*}
    Next, since $V_\varepsilon^2 < V_\varepsilon$
    and $$\sum_{k=0}^\varepsilon \binom{n}{k} \leq 2^{nh(\varepsilon/n)}$$ for $\varepsilon/n < 1/2$ (see~\cite{thomas2006elements}), the number $m_{\text{out}}$ of trials is lower bounded as follows:
    \begin{equation}
        \label{lower_outsider_independent_events}
        m_{\text{out}} \geq \frac{\log 2}{N (V_\varepsilon + V_\varepsilon^2)} \geq \frac{\log 2}{2N V_\varepsilon} \geq \frac{\log 2}{2N} 2^{n(1-h(\varepsilon/n))}
    \end{equation}

    \noindent
    For the upper bound, notice that for $\varepsilon/n < 1/2$, we have $$\sum_{k=0}^\varepsilon \binom{n}{k} \geq 2^{nh(\varepsilon/n)}/\sqrt{8\varepsilon(1-\varepsilon/n)}$$ (see~\cite{thomas2006elements}).
    Hence,
    \begin{equation}
        \label{upper_outsider_independent_events}
        m_\text{out} \leq \frac{\log 2}{N V_\varepsilon} \leq \frac{\log 2}{N} 2^{n(1-h(\varepsilon/n))} \times \sqrt{8\varepsilon(1-\varepsilon/n)}
    \end{equation}
\end{proof}

\paragraph{\textit{Case of different enrolled templates}}
In this case, a verification is performed on the enrolled templates to ensure that the database is only comprised of distinct templates. In other words, it is necessary to consider that templates are dependent, which is an essential change compared to the context of Theorem~\ref{db_size_outsider_independent_events}. It is then worth noting that $E_1^N$ cannot be seen as a union of independent events $E_1^1$ and that an exact measure of $E_1^N$ involves the cardinalities of multiple intersections of Hamming balls.
Therefore, the following result is a declination of Theorem~\ref{db_size_outsider_independent_events} in this specific context.

\begin{corollary}\label{db_size_outsider}
    Considering a similar setting of Theorem~\ref{db_size_outsider_independent_events} but with distinct templates in the database, then the median number $m_\text{out}$ of trials for the attacker to successfully impersonate a user is
    $$\Omega\left(2^{n(1-h(\varepsilon/n))-\log_2(N) -  \log_2 \left(1 + 6 \frac{N-1}{2^{n+1}}  \right) } \right)$$
    and
    $$O\left(2^{n(1-h(\varepsilon/n))  + \frac{1}{2} \log_2 \varepsilon \left( 1 - \frac{\varepsilon}{n} \right) - \log_2 N -  \log_2 \left( 1 + \frac{N-1}{2^{n+1}} \right) }\right).$$
\end{corollary}

\begin{proof}
    Similarly to the proof of Theorem~\ref{db_size_outsider_independent_events}, the probability of a success of a given trial is
    \begin{equation*}
        p = \mathbb{P}\left(t \in \bigcup\limits_{k=1}^N B_\varepsilon(v_k) \right) = 1 - \mathbb{P} \left( t \in \bigcap\limits_{k=1}^N \overline{B_\varepsilon(v_k)} \right)
    \end{equation*}
    where $t$ is the generated template of the attacker, and $v_k$ is the $k\text{-th}$ enrolled template in the template database.
    As templates $v_1, \dots, v_N$ are not independent, an alternative formulation can be expressed with conditional probabilities.
    Each conditional probability corresponds to the event that the $(\ell+1)^\text{th}$ enrolled template is sampled without replacement and that it does not be matched with~$t$:
    \begin{align*}
        \mathbb{P} \left( t \in \bigcap\limits_{k=1}^N \overline{B_\varepsilon(v_k)} \right) & =  \prod\limits_{\ell=0}^{N-1} \mathbb{P} \left( t \in \overline{B_\varepsilon(v_{\ell+1})}  \, \Big| \, t \in \bigcap\limits_{k=1}^\ell \overline{B_\varepsilon(v_{k})} \right) \\
                                                                                             & = \prod\limits_{\ell=0}^{N-1} \frac{2^n - \ell - |B_{\varepsilon}| }{2^n - \ell} = \prod\limits_{\ell=0}^{N-1} \left( 1 -  \frac{2^n}{2^n - \ell} V_\varepsilon \right)
    \end{align*}
    Next, since $-x-x^2 \leq \log(1-x) \leq -x$ for $0\leq x \leq \frac{1}{2}$ and
    if $N < 2^{n(1-h(\varepsilon/n))-1}$,
    then
    \begin{equation*}
        -  V_\varepsilon S_1 -  V_\varepsilon^2 S_2 \leq \log(1-p) \leq - V_\varepsilon S_1
    \end{equation*}
    where $S_1 = \sum_{\ell=0}^{N-1} \frac{2^n}{2^n - \ell}$ and $S_2 = \sum_{\ell=0}^{N-1} \left(\frac{2^n}{2^n - \ell}\right)^2$.
    Moreover, for what follows, $S_1$ and $S_2$ can be bounded as
    \begin{equation*}
        S_1 \geq N \left( 1 + \frac{N-1}{2^{n+1}}\right) \quad \text{ and } \quad S_2 \leq N \left( 1 + 6 \frac{N-1}{2^{n+1}}\right)
    \end{equation*}
    since $(1-x)^{-1} \geq 1 + x$ and  $(1-x)^{-2} \leq 1 + 6x$, for $0 \leq x \leq \frac{1}{2}$.
    Lastly, notice that $S_2> S_1$ and then the results follow.
\end{proof}
\noindent
Some care should be taken for the
choice of $\varepsilon$, since a high value for $\varepsilon$ dramatically reduces the
number of trials of the attacker, despite a large $n$. Theorem~\ref{db_size_outsider} makes clear the link between security parameters, hence allowing a safe choice of the parameters.

\subsubsection{$\kappa$-adaptive Attacks}

In order to study a case close to what might be a smart attacker, we consider that after a trial the attacker can infer that some of the remaining templates are actually unsuccessful templates. Next candidate guesses that are $\FMR$ away from tried templates are better choices than those near the center.
Such an inference should vary depending on the number of trials made and potentially exogenous information, and should lead to different amounts of inferred templates.
For the sake of simplicity, we investigate below an attacker model for which the attacker infers in average $\kappa$ unsuccessful templates.

\begin{definition}[$\kappa$-adaptive attacker]
    \label{def:attaquant_k-adaptatif}
    An attacker is $\kappa$-adaptive if for each of its trials to impersonate a user template she is able to identify $\kappa$ non-hit templates.
\end{definition}
As an example, a $0$-adaptive attacker tries to impersonate a user template by sampling into $\mathbb{Z}_2^n$ with replacement, and sampling without replacement for a $1$-adaptive attacker.
Proposition~\ref{prop:outsider_attack:negligible_difference_range} states that under reasonable conditions for the parameters, the number of tries required for a $\kappa$-adaptive attacker to succeed is equivalent to the required number of tries for a $0$-adaptive attacker.
Let $A(\kappa)$ denotes the number of trials of a $\kappa$-adaptive attacker to obtain a success.
From Proposition~\ref{prop:outsider_attack:negligible_difference_range}, we show that the median number of trials is equivalent for both attacker models.

\begin{theorem}
    \label{cor:naive_eq_k_adap:median_and_expect}
    Let $m_\kappa$ the median number of trials of a $k$-adaptive attacker to obtain a first success.
    If $p$ is negligible, $\kappa$ negligible compared to $\sqrt{2^n}$ and $m_\kappa \leq \sqrt{2^n}$,
    then, $m_0 \sim m_\kappa$.
\end{theorem}

\begin{proof}
    According to Proposition~\ref{prop:outsider_attack:negligible_difference_range},
    \begin{equation*}
        \frac{\mathbb{P} \left( A(0) \leq m_\kappa \right) }{\mathbb{P} \left( A(\kappa) \leq m_\kappa \right)} \xrightarrow[n\rightarrow +\infty]{} 1
    \end{equation*}
    and since $\mathbb{P} \left( A(\kappa) \leq m_\kappa \right) = 1/2$, we derive that $m_0$ tends to $m_\kappa$ as $n$ increases, if $m_\kappa \leq \sqrt{2^n}$.
\end{proof}

A $\kappa$-adaptive attacker is under realistic constraints equivalent to a $0$-adaptive attacker. Thus, for the sake of simplicity, considering Proposition~\ref{prop:outsider_attack:negligible_difference_range}, Theorem~\ref{cor:naive_eq_k_adap:median_and_expect} and  Proposition~\ref{prop:outsider_attack:negligible_difference}, in the remainder of the paper, we only derive theoretical results for a $0$-adaptive attacker model.
\paragraph*{\textit{Remark}} In practice, an attacker generates a template from the set $\mathbb{Z}_2^n$ deprived of the previously generated templates. When the number of trials (\textit{i.e.}, rounds) until the first success is low, the proposed simplified experiment above has a negligible bias.
Actually, the larger is $N$, the lesser is the number of trials until a first success. In a biometric database, the number $N$ of templates can be assumed large enough so that the cardinal of $\mathbb{Z}_2^n$  overwhelms the number of trials.


\paragraph{\textit{Intermediate results}}
We provide below complementary and intermediate results to Theorem~\ref{cor:naive_eq_k_adap:median_and_expect} that highlight the likenesses between a $0$-adaptive attacker and a $\kappa$-adaptive attacker, in terms of median number.
The first result (Proposition~\ref{prop:outsider_attack:negligible_difference_range}) and the second one (Proposition~\ref{prop:outsider_attack:negligible_difference}) give insight about likenesses between a $0$-adaptive attacker and a $\kappa$-adaptive attacker
expressed through probability mass functions and cumulative distribution functions.

Proposition~\ref{prop:outsider_attack:negligible_difference_range} states that under reasonable conditions for the parameters, the probability that a $\kappa$-adaptative attacker succeeds during the first $a^\text{th}$ trials is equivalent to the probability  of the same event for a $0$-adaptative attacker.

\begin{proposition}
    \label{prop:outsider_attack:negligible_difference_range}
    If $p$ is negligible 
    and $\kappa$ negligible compared to $\sqrt{2^n}$
    then, for a given number of trials $a\leq \sqrt{2^n}$, the probability that, among an amount of "$a$" trials, at least one successful trial of a $\kappa$-adaptive attacker is asymptotically equivalent to at least one successful trial of a $0$-adaptive attacker:
    \begin{equation*}
        \mathbb{P} \left( A(0) \leq a \right) \sim \mathbb{P} \left( A(\kappa) \leq a \right).
    \end{equation*}
\end{proposition}

\begin{proof}
    Let $p$ be a negligible probability, and notice that it is the case when the number of user templates is negligible compare to $2^n$.
    Then, $\kappa=g(n)\in o(2^{n/2})$ and $a\in\lbrace 1,\dots, 2^{n/2} \rbrace$.
    The probability of $\mathbb{P}(A(0)=a)$ is given by a geometric law of probability $p$. Thus, $\mathbb{P}(A(0) \leq a) = p\sum_{i=1}^a (1-p)^{i-1}$ and by using Lemma~\ref{lemma:insider_attack:distribution}, we have:
    \begin{align*}
        \mathbb{P}(A(\kappa) \leq a) = \sum\limits_{i=1}^a & \frac{p}{1 -(i-1)\frac{\kappa}{2^n}} \times \prod\limits_{j=1}^{i+2} \frac{ 1-p-j\frac{\kappa}{2^n}}{1-j\frac{\kappa}{2^n}}.
    \end{align*}
    When $n$ tends to infinity, according to the assumptions, the result follows since
    $$ \frac{\mathbb{P}(A(g(n)) \leq a)}{p}\xrightarrow[n\rightarrow +\infty]{} a \quad \text{ and }\quad \frac{\mathbb{P}(A(0) \leq a)}{p}\xrightarrow[n\rightarrow +\infty]{} a.$$
    Then, as the sums are finite, the result follows.
\end{proof}

Proposition~\ref{prop:outsider_attack:negligible_difference} states that under reasonable conditions for the parameters, the probability that a $\kappa$-adaptative attacker succeeds at the $a^\text{th}$ trial is equivalent to the probability that a $0$-adaptative attacker succeeds at the same trial.
\begin{proposition}
    \label{prop:outsider_attack:negligible_difference}
    If $p$ is negligible 
    and $\kappa$ negligible compared to $\sqrt{2^n}$
    then, for a given number of trials $a\leq \sqrt{2^n}$, the probability of a successful trial of a $\kappa$-adaptive attacker is asymptotically equivalent to the
    probability of successful trial of a $0$-adaptive attacker:
    \begin{equation*}
        \mathbb{P} \left( A(0) = a \right) \sim \mathbb{P} \left( A(\kappa) = a \right).
    \end{equation*}
\end{proposition}


\begin{proof}

    Let $p$ be a negligible probability, and notice that it is the case when the number of user templates is negligible compare to $2^n$. 
    Then, $\kappa=g(n)\in o(2^{n/2})$ and $a\in\lbrace 1,\dots, 2^{n/2} \rbrace$.
    Using Lemma~\ref{lemma:insider_attack:distribution_rate}, we have
    \begin{equation*}
        \frac{\mathbb{P} \left( A(0) = a \right)}{\mathbb{P} \left( A(\kappa) = a \right)} =   \prod_{j=1}^{a-1} \frac{1 - j\frac{\kappa}{2^n}}{ 1 - (j-1) \frac{\kappa}{2^n(1-p)}}.
    \end{equation*}
    When $n$ tends to infinity, according to the above assumptions, $\forall j \in \lbrace 1,\dots,a-1 \rbrace$
    $$1-j\frac{g(n)}{2^n} \xrightarrow[n\rightarrow +\infty]{} 1\quad \text{ and }\quad 1-(j-1)\frac{g(n)}{2^n(1-p)}\xrightarrow[n\rightarrow +\infty]{} 1.$$
    As all the terms of the product are positive, the result follows.
\end{proof}


\begin{lemma}
    \label{lemma:insider_attack:distribution_rate}
    The rate of probability of success at a given trial between an $0$-adaptive attacker and a $\kappa$-adaptive attacker is
    \begin{equation*}
        \frac{\mathbb{P} \left( A(0) = a \right)}{\mathbb{P} \left( A(\kappa) = a \right)} =   \prod_{j=1}^{a-1} \frac{1 - j\frac{\kappa}{2^n}}{ 1 - (j-1) \frac{\kappa}{2^n(1-p)}}
    \end{equation*}
    for $a \in \{ 1, \dots, \left\lceil \frac{2^n (1-p)}{\kappa} \right\rceil +1 \}$.
\end{lemma}

\begin{proof}
    Concerning the $0$-adaptive attacker, which corresponds to a sampler with replacement, the probability of a first success is a geometric distribution:
    \begin{equation*}
        \mathbb{P} \left( A(0) = a \right) = p (1-p)^{a-1}.
    \end{equation*}
    According to Lemma~\ref{lemma:insider_attack:distribution}, we have
    \begin{align*}
        \frac{\mathbb{P} \left( A(0) = a \right)}{\mathbb{P} \left( A(\kappa) = a \right)} & = (1-p)^{a-1} \frac{2^n -\kappa(a-1)}{2^n} \frac{\binom{\frac{2^n}{\kappa}}{a-1}}{\binom{\frac{2^n(1-p)}{\kappa}}{a-1}} \\
                                                                                           & = (1-p)^{a-1} \left( 1 - (a-1) \frac{\kappa}{2^n} \right)                                                               \\
                                                                                           & \quad \times \frac{\prod_{j=1}^{a-2} (1-j \frac{\kappa}{2^n})}{\prod_{j=0}^{a-2} (1-p-j \frac{\kappa}{2^n})}            \\
                                                                                           & =  \prod_{j=1}^{a-1}  \frac{ 1-j \frac{\kappa}{2^n}}{1-(j-1) \frac{\kappa}{2^n(1-p)}}.
    \end{align*}
\end{proof}

Lemma~\ref{lemma:insider_attack:distribution} gives the probability that the first success of a $\kappa$-adaptive attacker is the $a^\text{th}$ trial.

\begin{lemma}
    \label{lemma:insider_attack:distribution}
    The probability that the first success of a $\kappa$-adaptive attacker is the $a^\text{th}$ trial is given by
    \begin{align*}
        \mathbb{P}\Big( A(\kappa)=a \Big) = \frac{p2^n}{2^n-\kappa(a-1)} \frac{\binom{\frac{2^n(1-p)}{\kappa}}{a-1}}{\binom{\frac{2^n}{\kappa}}{a-1}}
    \end{align*}
    for $a \in \{ 1, \dots, \left\lceil \frac{2^n (1-p)}{\kappa} \right\rceil +1 \}$ and for $\kappa > 1$.
\end{lemma}
\begin{proof}

    If none of the $(a-1)^\text{th}$ trials lead to a success, then the probability that the $a^\text{th}$ trial fails, is
    \begin{align*}
        \mathbb{P} \left(A(\kappa) > a \, \big| \,  A(\kappa) > a-1 \right) & = \frac{2^n(1-p)-(a-1)\kappa}{2^n - (a-1)\kappa}                     \\
                                                                            & = \frac{\frac{2^n(1-p)}{\kappa} - (a-1)}{\frac{2^n}{\kappa} - (a-1)}
    \end{align*}
    then
    \begin{align*}
        \mathbb{P} \left(A(\kappa) = a \right) & = \Big(1-\mathbb{P} \left(A(\kappa) > a \, \big| \,  A(\kappa) > a-1 \right)\Big)                             \\
                                               & \phantom{\times} \times \prod_{j=1}^{a-1} \mathbb{P} \left(A(\kappa) > j \, \big| \,  A(\kappa) > j-1 \right) \\
                                               & = \frac{p2^n}{2^n -\kappa(a-1)} \times \frac{ \frac{2^n(1-p)}{\kappa}  }{\frac{2^n}{\kappa}}  \times  \dots   \\
                                               & \quad \times  \frac{ \frac{2^n(1-p)}{\kappa} -(a-2) }{\frac{2^n}{\kappa} - (a-2)}
    \end{align*}
    and the result follows since parts of the products above correspond to gamma function rates, which are related to binomial coefficients.
\end{proof}


\subsection{Multi-Insider Attacks}
\label{utg_attacks:in}

This section deals with an alternative case: all genuine users of the system are trying to impersonate others (see Theorem~\ref{thm:insider_attack:m_bounds}) or only a subset of them are trying to impersonate others (see Corollary~\ref{thm:insider_attack:m_bounds:l_attaquant}).
In these cases, each user generates templates in addition to their own enrolled template until one of them succeeds in impersonating any of the other $N-1$ users.
This attack consists of a certain number of rounds:
\begin{enumerate}
    \item Round $0$: The set of users do not generate any template, and try to impersonate any of the other users by using their own enrolled template. If none of them succeed, a new round begins: Round $1$.
    \item Round $1$, $\dots$, $m_\text{in}$: They all uniformly draw a random template and try to impersonate any of the other users with it. If none of them succeeds, a new round begins, and so on.
\end{enumerate}

The difference between the outsider scenario and the insider scenario is that the latter considers that all the users (or some of them) of the database are attackers, implying a higher cumulative number of guesses from all attackers. The total number of guesses is crucial in evaluating system security. Ultimately, in a large database of users, $m_{\text{in}}$ is invariably lower than $m_{\text{out}}$.

A main interest is then to determine, or at least approximate, the total number ($Nm_\text{in}$) of generated templates by the entire set of users until one of them succeeds with $50\%$ likelihood.
Notice that it is related to the number $m_\text{in}$ of rounds, which is approximated below, thanks to Theorem~\ref{thm:insider_attack:m_bounds}.

For what follows, $W$ denotes the event that a weak collision is found in a set of $N$ enrolled templates.
In the detailed experiment above, the enrolled templates of the database do not change. The probability of success of each individual is determined entirely by her coin tosses when generating a new template. With this in mind, an equivalent experiment can be considered where the templates of the database are randomly drawn from $\mathbb{Z}_2^n$ at each round.
We are then interested in this experiment that consists of a sequence of independent trials, so that each $W_j$ at a specific round~$j$ is independent to $W_i$ at another round~$i$.
Notice that for each round, the probability of success is the same, thus we denote $p_w$ the success probability during a round.
Furthermore, we recall that Round~$i$ only happens if Round 0 is a failure.
Then, a success during Round $i$, with $i>0$, is only given conditionally to $\overline{W_0}$, where $W_0$ denotes the event "a success at Round $0$".
Thus, $p_w=\mathbb{P}(W \, | \, \overline{W_0})$.
Since $W$ (and then $W_0$) is an event those it is difficult to derive probability computation, Proposition~\ref{prop:bound:nearcol} below provides lower and upper bounds of this probability.

\begin{proposition}
    \label{prop:bound:nearcol}
    For a database $\mathcal{B}$ of uniformly drawn templates, the probability of $W$ that there is a weak collision for at least two templates is bounded as follows
    \begin{align*}
        1 - \prod_{j=1}^N \Big( 1 - (j-1) V_\varepsilon \Big)                                                                        & \geq \mathbb{P}(W) \\
        1 - \prod\limits_{j=1}^N \Big( 1 - (j-1) V_\varepsilon + \frac{(j-1)(j-2)}{2}\mathcal{I}^{\varepsilon}_{\varepsilon+1} \Big) & \leq \mathbb{P}(W)
    \end{align*}
    where $V_{\varepsilon}$ is the measure of an $\varepsilon$-ball, and $\mathcal{I}_d^{\varepsilon}$ is the measure of the intersection of two $\varepsilon-$balls for which the Hamming distance between their centers is $d$.
\end{proposition}

\begin{proof}
    Consider $\mathcal{B} = (v_1, \dots, v_n)$ the database.
    Notice that $\mathbb{P}(W) = \mathbb{P} \Big( \exists u, v \in \mathcal{B}, u \in B_\varepsilon(v)\Big)  = 1 - \mathbb{P}\Big( \forall u, v \in \mathcal{B},  u \notin B_\varepsilon(v) \Big)$,
    and the term can be developed as:
    \begin{align}
        \label{eq:WC_chain_rule}
        \mathbb{P}\Big( \forall u,  v \in \mathcal{B},  u \notin B_\varepsilon(v) \Big) & = \prod\limits_{j=2}^N \mathbb{P}\Big( v_j \notin \bigcup\limits_{k=1}^{j-1} B_\varepsilon(v_k) \Big| \overline{W_{j-1}} \Big)
    \end{align}
    where $$\overline{W_j} = \text{"}v_2 \notin B_\varepsilon(v_1), \dots, v_{j-1} \notin \bigcup\limits_{k=1}^{j-1} B_\varepsilon(v_k)\text{"}.$$
    Then, according to Lemma~\ref{lemma:bounding_intermediate_proba}, each term of Equation~\eqref{eq:WC_chain_rule} can be bounded below by $1- (j-1)V_{\varepsilon}$ and above by $1- (j-1 ) V_{\varepsilon} + \frac{(j-1)(j-2)}{2}\mathcal{I}^{\varepsilon}_{\varepsilon+1}$.
\end{proof}
The previous proof involves a technical lemma (Lemma~\ref{lemma:bounding_intermediate_proba} given below) which is a key result in this paper.
Lemma~\ref{lemma:bounding_intermediate_proba} provides lower and upper bounds that are required to sequentially decompose the probability of a weak collision. The proof is in Appendix~\ref{sec:appendix:intermediate_result}.
\begin{lemma}
    \label{lemma:bounding_intermediate_proba}
    For a template database $\mathcal{B}= \{v_1, \dots, v_{N}\}$, the probability that an additional template $\tilde v$ matches a template of the database, according to a threshold $\varepsilon$, is bounded as:
    \begin{equation*}
        N V_{\varepsilon} -  \mathcal{I}^{\varepsilon}_{\varepsilon+1} N(N-1)/2 \leq \mathbb{P}\left(\tilde v \in \cup_{k=1}^{N} B_{\varepsilon}(v_k) \, \big| \, \overline{W_N}\right) \leq NV_{\varepsilon}
    \end{equation*}
    where $\overline{W_N}$ denotes the event "$v_2 \notin B_{\varepsilon}(v_1) , \dots, v_{N} \notin \cup_{k=1}^{N-1}  B_{\varepsilon}(v_k)$" .
\end{lemma}

Having characterized the probability of a weak collision, the security can now be evaluated when all users of the database are potential attackers. The estimated bounds for the median number of rounds are given in Theorem~\ref{thm:insider_attack:m_bounds} and it is translated in Corollary~\ref{thm:insider_attack:m_bounds:l_attaquant} to fit an intermediate setting. The proof is in the Appendix~\ref{sec:appendix:intermediate_result}.

\begin{theorem}
    \label{thm:insider_attack:m_bounds}
    Let $N$, $n$ and $\varepsilon$ be fixed parameters with $\varepsilon/n \leq 1/2$. The median number $m_\text{in}$ of rounds $m$ until at least one of the users successfully impersonates another one is
    $$\Omega\left(2^{n(1-h(\varepsilon/n))-3\log_2(N)}\right) \text{ and }
        O\left(2^{n(1-h(\varepsilon/n)+\alpha)-2\log_2(N)}\right)$$
    where $h(\cdot)$ is the binary entropy function, and with $\alpha < h(\varepsilon/n)$ and if the following holds
    \begin{equation}
        \label{eq:condition_sur_N}
        N \leq 2 + 2^{-\varepsilon} \left( \frac{1-\varepsilon/n}{\varepsilon/n}\right)^{\left\lceil \frac{\varepsilon+1}{2} \right\rceil} \left( \frac{1}{\sqrt{8\varepsilon(1-\varepsilon/n)}}-2^{-n\alpha} \right).
    \end{equation}
\end{theorem}

The assumption that all users in the database are attackers may seem too strong. For a more realistic evaluation, it is preferable to perform the evaluation assuming that there is a few number of attackers in the system.
Corollary~\ref{thm:insider_attack:m_bounds:l_attaquant} gives the bounds on the median number of rounds when only a subset of clients are attackers.

\begin{corollary}
    \label{thm:insider_attack:m_bounds:l_attaquant}
    Assuming the same setting as in Thoerem~\ref{thm:insider_attack:m_bounds} but with only $\ell$ attacking users among $N$ aiming at impersonating another one, the median number $m_\text{in}$ of rounds $m$ until at least one successfully impersonates another a user template is
    $$\Omega(2^{n(1-h(\varepsilon/n))-2\log_2(N)-\log_2(\ell)})$$
    and
    $$ O(2^{n(1-h(\varepsilon/n)+\alpha)-\log_2(N)-\log_2(\ell)}).$$
\end{corollary}

\begin{proof}
    In this case, the probability that at least one of the $\ell$ user succeeds is:
    \begin{equation*}
        p_w = 1 - \prod_{k=1}^\ell (1-p_k)
    \end{equation*}
    Following the same steps as in the proof of Theorem~\ref{thm:insider_attack:m_bounds}, lower and upper bounds of the median number are:
    \begin{align*}
        \left\lceil \frac{ (\log 2) 2^{n \left(1-h(\varepsilon/n)  \right) }}{\ell N(N-1) } \right\rceil \leq m_\text{in} \leq  \left\lceil \frac{(\log 2) \, 2^{n ( 1-h(\varepsilon/n) + \alpha )}}{\ell  (N-1) }
        \right\rceil
    \end{align*}
    and the result follows.
\end{proof}

In the proof of Theorem~\ref{thm:insider_attack:m_bounds}, the measure of the intersection $\mathcal{I}_{\varepsilon+1}^\varepsilon$ is bounded according to Lemma~\ref{lemma:entropyballintersection} in order to highlight the relationship with the binary entropy function $h(\cdot)$.
Its statement is given below and its proof is in Appendix~\ref{sec:appendix:balls_intersection}.
\begin{lemma}
    \label{lemma:entropyballintersection}
    For $u$ and $v \in \mathbb{Z}^n_2$, $d = d_\mathcal{H}(u,v)$, with $d \leq 2\varepsilon$, $\varepsilon/n < 1/2$, then
    \begin{align*}
        \mathcal{I}_d^\varepsilon \leq  2^{n \left( h\left( \frac{\varepsilon-\lceil d/2 \rceil}{n}\right) -1\right) + d} \leq 2^{n \left(h(\varepsilon/n)-1\right)+d} \left(\frac{\varepsilon/n}{1-\varepsilon/n} \right)^{\left\lceil d/2 \right\rceil}.
    \end{align*}
\end{lemma}

\section{Experimental Results}

\begin{table*}[]
    \resizebox*{\textwidth}{!}{%
        \begin{tabular}{c|cccccccccccc|cccccccccccc}
            \toprule
            $n$                & \multicolumn{12}{c|}{$128$} & \multicolumn{12}{c}{$256$}                                                                                                                                                                                                                                                                                                                                                                                                                                                                                                                                                                                                                                                                                                                                                                                                                                                                               \\ \midrule
            $N$ ($log_{10}$)   & \multicolumn{4}{c|}{$4$}    & \multicolumn{4}{c|}{$6$}   & \multicolumn{4}{c|}{$8$}                  & \multicolumn{4}{c|}{$4$}                    & \multicolumn{4}{c|}{$6$}                  & \multicolumn{4}{c}{$8$}                                                                                                                                                                                                                                                                                                                                                                                                                                                                                                                                                                                                                                                                                                               \\ \midrule
            $\varepsilon$      & $12$                        & $19$                       & $25$                                      & \multicolumn{1}{c|}{$51$}                   & $12$                                      & $19$                                      & $25$                                      & \multicolumn{1}{c|}{$51$}                   & $12$                                      & $19$                                       & $25$                    & $51$                                        & $12$                                       & $19$                                       & $25$                     & \multicolumn{1}{c|}{$51$}                    & $12$                                       & $19$                     & $25$                     & \multicolumn{1}{c|}{$51$}                    & $12$                     & $19$                     & $25$                     & $51$                    \\ \midrule
            $h(\varepsilon/n)$ & $0.4$                       & $0.6$                      & $0.7$                                     & \multicolumn{1}{c|}{$1.0$}                  & $0.4$                                     & $0.6$                                     & $0.7$                                     & \multicolumn{1}{c|}{$1.0$}                  & $0.4$                                     & $0.6$                                      & $0.7$                   & $1.0$                                       & $0.3$                                      & $0.4$                                      & $0.5$                    & \multicolumn{1}{c|}{$0.7$}                   & $0.3$                                      & $0.4$                    & $0.5$                    & \multicolumn{1}{c|}{$0.7$}                   & $0.3$                    & $0.4$                    & $0.5$                    & $0.7$                   \\ \midrule
            Lower bound        & \multirow{2}{*}{$ 57 $}     & \multirow{2}{*}{$ 37 $}    & \multirow{2}{*}{$ 23 $}                   & \multicolumn{1}{c|}{\multirow{2}{*}{$ 0 $}} & \multirow{2}{*}{$ 50 $}                   & \multirow{2}{*}{$ 30 $}                   & \multirow{2}{*}{$ 16 $}                   & \multicolumn{1}{c|}{\multirow{2}{*}{$ 0 $}} & \multirow{2}{*}{$ 43 $}                   & \multirow{2}{*}{$ 23 $}                    & \multirow{2}{*}{$ 10 $} & \multicolumn{1}{c|}{\multirow{2}{*}{$ 0 $}} & \multirow{2}{*}{$ 172 $}                   & \multirow{2}{*}{$ 145 $}                   & \multirow{2}{*}{$ 124 $} & \multicolumn{1}{c|}{\multirow{2}{*}{$ 58 $}} & \multirow{2}{*}{$ 166 $}                   & \multirow{2}{*}{$ 138 $} & \multirow{2}{*}{$ 117 $} & \multicolumn{1}{c|}{\multirow{2}{*}{$ 51 $}} & \multirow{2}{*}{$ 159 $} & \multirow{2}{*}{$ 131 $} & \multirow{2}{*}{$ 111 $} & \multirow{2}{*}{$ 45 $} \\
            Outsider ($log_2$) &                             &                            &                                           & \multicolumn{1}{c|}{}                       &                                           &                                           &                                           & \multicolumn{1}{c|}{}                       &                                           &                                            &                         &                                             &                                            &                                            &                          & \multicolumn{1}{c|}{}                        &                                            &                          &                          & \multicolumn{1}{c|}{}                        &                          &                          &                          &                         \\ \midrule
            Upper bound        & \multirow{2}{*}{$ 58 $}     & \multirow{2}{*}{$ 39 $}    & \multirow{2}{*}{$ 25 $}                   & \multicolumn{1}{c|}{\multirow{2}{*}{$ 0 $}} & \multirow{2}{*}{$ 52 $}                   & \multirow{2}{*}{$ 32 $}                   & \multirow{2}{*}{$ 19 $}                   & \multicolumn{1}{c|}{\multirow{2}{*}{$ 0 $}} & \multirow{2}{*}{$ 45 $}                   & \multirow{2}{*}{$ 25 $}                    & \multirow{2}{*}{$ 12 $} & \multicolumn{1}{c|}{\multirow{2}{*}{$ 0 $}} & \multirow{2}{*}{$ 174 $}                   & \multirow{2}{*}{$ 147 $}                   & \multirow{2}{*}{$ 126 $} & \multicolumn{1}{c|}{\multirow{2}{*}{$ 60 $}} & \multirow{2}{*}{$ 167 $}                   & \multirow{2}{*}{$ 140 $} & \multirow{2}{*}{$ 120 $} & \multicolumn{1}{c|}{\multirow{2}{*}{$ 54 $}} & \multirow{2}{*}{$ 161 $} & \multirow{2}{*}{$ 133 $} & \multirow{2}{*}{$ 113 $} & \multirow{2}{*}{$ 47 $} \\
            Outsider ($log_2$) &                             &                            &                                           & \multicolumn{1}{c|}{}                       &                                           &                                           &                                           & \multicolumn{1}{c|}{}                       &                                           &                                            &                         &                                             &                                            &                                            &                          & \multicolumn{1}{c|}{}                        &                                            &                          &                          & \multicolumn{1}{c|}{}                        &                          &                          &                          &                         \\ \midrule
            Lower bound        & \multirow{2}{*}{$31$}       & \multirow{2}{*}{$11$}      & \multicolumn{2}{c|}{\multirow{2}{*}{$0$}} & \multirow{2}{*}{$11$}                       & \multicolumn{3}{c|}{\multirow{2}{*}{$0$}} & \multicolumn{4}{c|}{\multirow{2}{*}{$0$}} & \multirow{2}{*}{$146$}                    & \multirow{2}{*}{$118$}                      & \multirow{2}{*}{$98$}                     & \multicolumn{1}{c|}{\multirow{2}{*}{$32$}} & \multirow{2}{*}{$126$}  & \multirow{2}{*}{$99$}                       & \multirow{2}{*}{$78$}                      & \multicolumn{1}{c|}{\multirow{2}{*}{$12$}} & \multirow{2}{*}{$106$}   & \multirow{2}{*}{$79$}                        & \multirow{2}{*}{$58$}                      & \multirow{2}{*}{$0$}                                                                                                                                                                                          \\
            Insider ($log_2$)  &                             &                            & \multicolumn{2}{c|}{}                     &                                             & \multicolumn{3}{c|}{}                     & \multicolumn{4}{c|}{}                     &                                           &                                             &                                           & \multicolumn{1}{c|}{}                      &                         &                                             &                                            & \multicolumn{1}{c|}{}                      &                          &                                              &                                            &                                                                                                                                                                                                               \\ \midrule
            Upper bound        & \multirow{2}{*}{$44$}       & \multirow{2}{*}{$24$}      & \multirow{2}{*}{$10$}                     & \multicolumn{1}{c|}{\multirow{2}{*}{$0$}}   & \multirow{2}{*}{$31$}                     & \multirow{2}{*}{$11$}                     & \multicolumn{2}{c|}{\multirow{2}{*}{$0$}} & \multirow{2}{*}{$17$}                       & \multicolumn{3}{c|}{\multirow{2}{*}{$0$}} & \multirow{2}{*}{$160$}                     & \multirow{2}{*}{$132$}  & \multirow{2}{*}{$111$}                      & \multicolumn{1}{c|}{\multirow{2}{*}{$45$}} & \multirow{2}{*}{$146$}                     & \multirow{2}{*}{$118$}   & \multirow{2}{*}{$98$}                        & \multicolumn{1}{c|}{\multirow{2}{*}{$32$}} & \multirow{2}{*}{$133$}   & \multirow{2}{*}{$105$}   & \multirow{2}{*}{$85$}                        & \multirow{2}{*}{$18$}                                                                                    \\
            Insider ($log_2$)  &                             &                            &                                           & \multicolumn{1}{c|}{}                       &                                           &                                           & \multicolumn{2}{c|}{}                     &                                             & \multicolumn{3}{c|}{}                     &                                            &                         &                                             & \multicolumn{1}{c|}{}                      &                                            &                          &                                              & \multicolumn{1}{c|}{}                      &                          &                          &                                              &                                                                                                          \\ \bottomrule
        \end{tabular}
    }
    \caption{Bounds for the number of operations of both an insider and an outsider, in function of $n$, $N$ and $\varepsilon$.}
    \label{Insider_tab}
\end{table*}

\label{template_sec}
This section provides some experimental results as well as security metrics related to the theoretical bounds given in this paper.
First, Proposition~\ref{prop:outsider_attack:negligible_difference_range} is tested on reasonable scenarios to highlight the equivalence of the two attacker models. Then, the numbers of guesses for both an attacker in the insider case (Theorem~\ref{thm:insider_attack:m_bounds}) and the outsider case (Theorem~\ref{db_size_outsider}) are investigated. In the end, we propose new security metrics to evaluate the resilience of a database with respect to near collisions.

\subsection{Numerical Evaluations: Databases Security}

Preferably, a near-collision should only occur for legitimate users on their
own enrolled template, in reference to intra-class near-collisions.
However, as shown in the previous sections, inter-class near-collisions occur
well below the birthday bound of a cryptographic hash function.
Near-collisions negatively impact biometric systems in two ways. The first is
the increase in FMR and the second is the onset of master templates that
facilitates multiple impersonations~\cite{durbet2022near}.
To limit these near-collisions, the system can act on $3$ parameters: $n$, $N$
and $\varepsilon$.

The analyses focus on a database $\mathcal{B}$ of
uniformly distributed template in $\mathbb{Z}_2^n$, \textit{i.e.}, when the
biometric transformation acts like a perfect random function.
This enables to provide an upper bound on run-time complexities. Although the
assumption of a uniform distribution yields an overestimated upper bound, it is
helpful for securely parametrizing the transformation scheme.
However, concerning the lower bound, it is above reality if the distribution is non-uniform, as it is the case for deterministic biometric transformation.
Weak near-collisions occur more frequently in the case of a skewed distribution.


To evaluate the security of a database based on these results, a first security
score against a passive attacker, denoted $S_1$, is introduced.
This score is given for $\mathcal{B}$ and its parameters.
By denoting $p_1$ the upper bound of the probability
of a weak near-collision, the score is $1-p_1$ and is valued in $[0,1]$.
The closer to~$1$ the score is, the more resilient is the database.
A database
$\mathcal{B}$ is said to be resilient to passive attacks iff $S_1$ is above
$1/2$.
As an example, the database defined by ($n=64$, $N=50$, $\varepsilon=15$) with
$S_1=0.9852$ is considered robust, while the database defined by ($n=64$,
$N=50$, $\varepsilon=19$) with $S_1=0.3792$ is not.

They each show the variation of this probability when
$\varepsilon$, $N$ and $n$ vary respectively. As expected, it
increases with $\varepsilon$ or $N$, but decreases with $n$.
For $n=128$ and $N=100$, the robustness is obtained by taking
$\varepsilon \leq 43$.
For $\varepsilon=12$ and $N=100$, the robustness is obtained for a template
dimension $n \ge 51$.
For $\varepsilon=18$ and $n=64$, the robustness is achieved for a number of
users $N \leq 67$.

Other indicators of security of biometric databases is the resilience to attacks
in the outsider and insider scenarios. This is closely related to Theorem~\ref{db_size_outsider} and the
inequations~\eqref{lower_outsider_independent_events} and~\eqref{upper_outsider_independent_events}, and Theorem~\ref{thm:insider_attack:m_bounds} and the
inequations~\eqref{Min_insider_bound} and~\eqref{Max_insider_bound} respectively.
In Table~\ref{Insider_tab}, some parameters of $n$, $N$ and $\varepsilon$ give a poor resistance to those attacks.
Two other scores $S_2$ and $S_3$ are introduced accordingly. By denoting $p_2$
and $p_3$ the lower bounds for the number of rounds of an outsider and an
insider respectively, the corresponding scores are $S_2=\log_2(p_2/2^{128})$ and
$S_3=\log_2(p_3/2^{128})$.
The database is resilient to these attacks if the scores are greater than or equal to $0$.
According to Table~\ref{Insider_tab}, the triplet ($n=256$, $N=10^4$,
$\varepsilon=19$) yields the scores $S_2=44$ and $S_3=18$. However, the
triplet ($n=128$, $N=10^4$, $\varepsilon=12$) yields the low scores $S_2=-71$
and $S_3=-97$.
The table shows large security drops for the less common insider
scenario (\textit{e.g.}, $n=256$, $N=10^8$, $\varepsilon=19$), so that we could
relax the constraint on $S_3$.

Usually, biometric recognition systems are parameterized by adjusting a threshold using a training dataset in experimental evaluations to achieve the Equal Error Rate ({\tt EER}). However, the obtained threshold could be too large, whence not providing the expected level of security based on our analyses. Therefore, a trade-off has to be found between the False Non-Matching Rate ({\tt FNMR}, with respect to the {\tt EER})  and the above security scores.

\section{Concluding Remarks}
\label{conclusion}

Firstly, we studied the case of untargeted attacks and the probability of a near-collision occurrence based on system measurements. This analysis allowed us to highlight the lack of security and accuracy in current systems.
Regarding untargeted attacks and their complexities, the security level of the studied systems does not exceed $16$ bits, whereas a system is considered secure if it provides at least between $128$ and $256$ bits of security (see Table~\ref{Table_FMR_result}). Such security is achievable in theory as shown in the metric space analysis (see Table~\ref{Insider_tab}).
On the other hand, having a low probability of near collision ensures the accuracy of the system, especially for identification systems.
Most of the studied recognition systems have a high near-collision probability (see Table~\ref{Table_FMR_result}).
This is due to the excessively high number of users in their databases. To mitigate this problem, we have provided a method to compute the maximum number of enrolled users to preserve a given security level.

Secondly, we explored the limits these systems could have by modeling them with a metric space.
Under the same assumptions (active attackers with the minimum of information leakage from the system), our results fall in line with a previous work of Pagnin \textit{et al.}~\cite{PagninDAM14} on targeted attacks.
As untargeted attacks are common in password cracking, we have studied this case which was not formally investigated in the literature.
We have presented two biometric recovery attacks regardless of the modalities and provided their complexities in Big-Omicron and Big-Omega as well as experimentations to support those results (see Table~\ref{Insider_tab}).
The results could be simplified thanks to a simpler but equivalent attacker model (see Table~\ref{Tab:equiv_ratio}). Indeed, going from the $\kappa$-adaptive model to the naive attacker leads to a lot of simplifications.
Our results highlight the importance of the choice of the parameters of a database and provide a way to pick them carefully.
We provide a new adaptive security metric based on those attacks and we investigate the probability of a weak near-collision and of a master template occurrence.

Future research directions would be to broaden our results, by considering non-binary templates and other distances.



\section*{Acknowledgement}
The authors acknowledges the support of the French Agence Nationale de la Recherche (ANR), under grant ANR-20-CE39-0005 (project PRIVABIO).

\bibliographystyle{abbrv}
\bibliography{biblio}

\section*{Appendices}

\section{Intermediate results}
\label{sec:appendix:intermediate_result}

In the following, $\mathcal{B} = (v_1, \dots, v_N)$ denotes a template database and each $v \in \mathcal{B}$ is assumed to be uniformly drawn in $\mathbb{F}^n_2$ and independently from each other:  $v \overset{\text{ind}}{\sim} \operatorname{Unif}(\mathbb{F}_2^n)$.
For what follows, we recall that $V_{\varepsilon} = \frac{1}{2^n} \sum_{k=0}^\varepsilon \binom{n}{k}$ and the measure of the intersection of two $\varepsilon-$balls for which the Hamming distance between their centers is $d$, is given by:
\begin{equation*}
    \mathcal{I}^\varepsilon_d = \frac{1}{2^n} \sum\limits_{k=\max(0,d-\varepsilon)}^{\min(\varepsilon,d)} \sum\limits_{i=0}^{\min(\varepsilon-k,\varepsilon-d+k)} \binom{d}{k} \binom{n-d}{i}.
\end{equation*}
Below, we also rely on the following inequality, if $k/n < 1/2$:
\begin{equation}
    \label{eq:approximation_binomial_sum}
    \frac{2^{nh(k/n)}}{\sqrt{8k (1-k/n)}} \leq \sum_{j=0}^k \binom{n}{j} \leq 2^{nh(k/n)}
\end{equation}

Concerning the upper bound of Lemma~\ref{lemma:bounding_intermediate_proba}, notice that it can be larger than 1, which is not relevant.
Furthermore, notice that it implies the same conclusion for the upper bound of Proposition~\ref{prop:bound:nearcol}.
\begin{lemma}
    \label{lemma:condition_N_upper_bound}
    If $N < 2^{n(1-h(\varepsilon/n))}$,  $\varepsilon/n < 1/2$, then $N V_{\varepsilon}<1$.
\end{lemma}
\begin{proof}
    First, remark that  $N < 2^{n(1-h(\varepsilon/n))}$ is equivalent to $N2^{n(h(\varepsilon/n)-1)} < 1$ and since $\varepsilon/n < 1/2$,
    $NV_\varepsilon \leq N2^{n(h(\varepsilon/n)-1)}$, and the result follows.
\end{proof}
Note also that the  lower bound of Lemma~\ref{lemma:bounding_intermediate_proba} can be negative if $N$ is large.
Below Lemma~\ref{lemma:condition_N_lower_bound} provides a threshold above which this lower bound is not informative.
\begin{lemma}
    \label{lemma:condition_N_lower_bound}
    If $N < N_\text{max}$, then $N V_{\varepsilon} - \frac{N(N-1)}{2} \mathcal{I}^{\varepsilon}_{\varepsilon+1} > 0$, where
    \begin{equation*}
        N_{max} = 1 + \frac{2^{-\varepsilon}}{\sqrt{8\varepsilon(1-\varepsilon/n)}} \left( \frac{1-\varepsilon/n }{\varepsilon/n}\right)^{\left\lceil \frac{\varepsilon+1}{2}\right\rceil}
    \end{equation*}
\end{lemma}
\begin{proof}
    First notice that $N V_{\varepsilon} - \frac{N(N-1)}{2} \mathcal{I}^{\varepsilon}_{\varepsilon+1} > 0$ holds if $N < 1 + 2 V_\varepsilon/\mathcal{I}_{\varepsilon+1}^\varepsilon$.
    By using Lemma~\ref{lemma:entropyballintersection} and \eqref{eq:approximation_binomial_sum} for $\varepsilon/n<  1/2$, $N_{max} \leq  1 + 2 V_\varepsilon/\mathcal{I}_{\varepsilon+1}^\varepsilon$,
    and the result follows.
\end{proof}
\noindent Next, Lemma~\ref{lemma:condition_N_lower_bound_inf_1} provides constraints on $n$ and $\varepsilon$ so that the  lower bound of Lemma~\ref{lemma:bounding_intermediate_proba} is always lower than 1.
\begin{lemma}
    \label{lemma:condition_N_lower_bound_inf_1}
    The equation $N V_{\varepsilon} - \frac{N(N-1)}{2} \mathcal{I}^{\varepsilon}_{\varepsilon+1} < 1$ holds if $\varepsilon/n < 1/2$ and $n(1-2h(\varepsilon/n)) >  2\log_2(3) -\varepsilon$.
\end{lemma}
\begin{proof}
    Write the quadratic function as
    \begin{equation*}
        f(N) = -N^2 \mathcal{I}_{\varepsilon+1}^\varepsilon/2 + N(V_\varepsilon + \mathcal{I}_{\varepsilon+1}^\varepsilon/2)
    \end{equation*}
    and determine that the quadratic function in $N$ is bounded above by:
    \begin{equation*}
        f\Big( (V_\varepsilon + \frac{1}{2}\mathcal{I}_{\varepsilon+1}^\varepsilon) / \mathcal{I}_{\varepsilon+1}^\varepsilon \Big) = (V_\varepsilon + \frac{1}{2} \mathcal{I}_{\varepsilon+1}^\varepsilon)^2 / 2 \mathcal{I}_{\varepsilon+1}^\varepsilon < \frac{9}{2} \frac{V_\varepsilon^2}{\mathcal{I}_{\varepsilon+1}^\varepsilon}
    \end{equation*}
    with $\mathcal{I}_{\varepsilon+1}^{\varepsilon} < V_{\varepsilon}$ since the intersection of two $\varepsilon$-balls is lower than an $\varepsilon$-ball as soon as the Hamming distance $d$ between theirs centers is different from $0$.
    According to Lemma~\ref{lemma:entropyballintersection_min_epsilon_p_1},
    \begin{equation*}
        \frac{9}{2} \frac{V_\varepsilon^2}{\mathcal{I}_{\varepsilon+1}^\varepsilon} \leq \frac{9}{2} \frac{2^{-2n(1-h(\varepsilon/n))}}{2^{-n} \left( 2^\varepsilon -1 \right)} = \frac{9}{2^{\varepsilon}} 2^{-n(1-2h(\varepsilon/n))}
    \end{equation*}
    since $\varepsilon >0$ so that $2^\varepsilon-1 \leq 2^{\varepsilon-1}$.
    Now, observe that
    \begin{equation*}
        9 \times 2^{-\varepsilon} 2^{-n(1-2h(\varepsilon/n))} <1  \, \Leftrightarrow \,
        n(1-2h(\varepsilon/n)) > 2\log_2(3) -\varepsilon
    \end{equation*}
\end{proof}
One can see that $h(1/10) \approx 0.47$, therefore if $\varepsilon \geq 4$ and $\varepsilon/n < 1/10$ then $n(1-2h(\varepsilon/n)) >  2\log_2(3) -\varepsilon$.
Moreover notice for the same reason as Lemma~\ref{lemma:condition_N_lower_bound} that the lower bound of Proposition~\ref{prop:bound:nearcol} decreases for large value of $N$.
However, one can expect that the probability of a weak collision increases as $N$ becomes larger.
Then, the aforementioned lower bound is not relevant for large value of $N$.
For what follows, $b_N$ denotes the lower bound of the probability of a weak collision.
\begin{proposition}
    \label{prop:borne_inf_decroissante}
    The lower bound $b_N$ is increasing if $N < N_\text{max}$, where $N_{max} = 1 + \frac{2^{-\varepsilon}}{\sqrt{8\varepsilon(1-\varepsilon/n)}} \left( \frac{1-\varepsilon/n }{\varepsilon/n}\right)^{\left\lceil \frac{\varepsilon+1}{2}\right\rceil}$
    and if $n(1-2h(\varepsilon/n)) >  2\log_2(3) -\varepsilon$.
\end{proposition}
\begin{proof}
    Recall that:
    \begin{equation*}
        b_N = 1 - \prod_{j=1}^N \left( 1 - (j-1) V_\varepsilon + \frac{(j-1)(j-2)}{2} \mathcal{I}_{\varepsilon+1}^\varepsilon \right).
    \end{equation*}
    Consider the difference $b_{N+1} - b_N = \left(N V_\varepsilon - \frac{N(N-1)}{2} \mathcal{I}_{\varepsilon+1}^{\varepsilon} \right) b_N$.
    According to Lemma~\ref{lemma:condition_N_lower_bound}, $N V_\varepsilon - \frac{N(N-1)}{2} \mathcal{I}_{\varepsilon+1}^{\varepsilon} > 0$ since $N < N_\text{max}$ and $1 - (j-1) V_\varepsilon + \frac{(j-1)(j-2)}{2} \mathcal{I}_{\varepsilon+1}^\varepsilon > 0$.
    according to Lemma~\ref{lemma:condition_N_lower_bound_inf_1} if $n(1-2h(\varepsilon/n)) >  2\log_2(3) -\varepsilon$. Therefore, $b_{N+1} - b_N>0$ and then $b_N$ is increasing if constraints on $N,n$ and $\varepsilon$ hold.
\end{proof}

Proof of the theorem~\ref{thm:insider_attack:m_bounds}.
\begin{proof}
    As the result focuses on the number of round(s) until a success, it suggests that none of the users succeeds at impersonating another one during Round 0.
    It means that the event considered below is given that there is no near-collision at Round 0, which corresponds to Event $\overline{W_0}$.
    At each round, the $k^{\text{th}}$ user generates  a new template $v_k'$, independently of other users and independently to its own template, and the probability that at a given round this user is successfully authenticated as another one is given by:
    \begin{equation*}
        p_k = \mathbb{P}\left(v_k' \in \bigcup\limits_{j\neq k}^N B_\varepsilon(v_j) \, \Big| \, \overline{W_0}\right).
    \end{equation*}
    Then, the probability that at least one user succeeds is:
    \begin{align*}
        p_w & = \mathbb{P}\left( \exists v_k' \text{ such that } v_k' \in \bigcup\limits_{j\neq k}^N B_\varepsilon(v_j) \, \Big| \, \overline{W_0} \right)                  \\
            & = 1 - \prod_{k=1}^N \mathbb{P}\left( v_k' \notin \bigcup\limits_{j\neq k}^N B_\varepsilon(v_j) \, \Big| \, \overline{W_0} \right) = 1 - \prod_{k=1}^N (1-p_k)
    \end{align*}
    and furthermore
    \begin{align*}
        \log(1-p_w) = \sum_{k=1}^N \log(1-p_k).
    \end{align*}
    Since each round is independent to other rounds, and that the probability of success is the same at each round, the number of round(s) until a user succeeds follows a geometric distribution.
    For a geometric distribution of probability $p_w$, the median is given by
    \begin{equation*}
        m_\text{in} = \left\lceil \frac{-\log 2}{\log (1-p_w)} \right\rceil.
    \end{equation*}
    According to Lemma~\ref{lemma:bounding_intermediate_proba} it follows that
    \begin{align}
        \label{eq:median_bounding_log}
        \log(1-p_w) & \geq N \log\left(1- (N-1) V_\varepsilon\right)                                                                       \\
        \log(1-p_w) & \leq N \log \left(1-(N-1) V_\varepsilon + \frac{(N-1)(N-2)}{2} \mathcal{I}_{\varepsilon+1}^\varepsilon\right) \notag
    \end{align}
    which is required to compute lower and upper bounds of $m_\text{in}$.
    Then, the lower bound is obtained according to \eqref{eq:median_bounding_log}:
    \begin{align*}
        m_\text{in} & \geq \left\lceil \frac{-\log 2}{N \log (1-(N-1) V_\varepsilon)} \right\rceil
        \\
                    & \geq \left\lceil \frac{-\log 2}{
            N \log \left(
            1-(N-1)  2^{n \left(
                    h(\varepsilon/n) - 1
                    \right) }
            \right)
        }    \right\rceil                                                                                                                                                 \\
                    & \geq \left\lceil \frac{ \log 2 }{N(N-1) 2^{n \left(h(\varepsilon/n) -1 \right)} + N(N-1)^2 2^{2n \left(h(\varepsilon/n) -1 \right)  }} \right\rceil
    \end{align*}
    since $ \log (1-p) \geq -p-p^2$. Next, notice that $2( h(\varepsilon/n) -1) < ( h(\varepsilon/n) -1)$ since $ \varepsilon/n < 1/2$, then
    \begin{equation}
        \label{Min_insider_bound}
        m_\text{in} \geq \left\lceil \frac{ (\log 2) 2^{n \left(1-h(\varepsilon/n)  \right) }}{N^2(N-1) } \right\rceil.
    \end{equation}
    Concerning the upper bound, consider the upper bound of~\eqref{eq:median_bounding_log}:
    \begin{align*}
        m_\text{in} & \leq  \left\lceil \frac{\log 2}{N (N-1) \left( V_\varepsilon - \frac{(N-2)}{2} \mathcal{I}_{\varepsilon+1}^\varepsilon \right)} \right\rceil
    \end{align*}
    since $ \log (1-p) \leq -p$.
    By using a lower bound for $V_\varepsilon$ and Lemma~\ref{lemma:entropyballintersection}, and by denoting $c = \left\lceil \frac{\varepsilon +1 }{2}  \right\rceil$ to ease the~inequations:
    \begin{align*}
        m_\text{in}
         & \leq \left\lceil \frac{(\log 2) 2^{n ( 1-h(\varepsilon/n) )}}{N (N-1) \left(
            \frac{
                1
            }{
                \sqrt{8\varepsilon (1-\varepsilon/n)}
            }
            - (N-2) \left( \frac{\varepsilon/n}{1- \varepsilon/n} \right)^c 2^{\varepsilon }
            \right)
        }
        \right\rceil.
    \end{align*}
    Next, according to \eqref{eq:condition_sur_N}
    \begin{equation*}\frac{
            1
        }{
            \sqrt{8\varepsilon (1-\varepsilon/n)}
        }
        - (N-2) \left( \frac{\varepsilon/n}{1- \varepsilon/n} \right)^{\left\lceil \frac{\varepsilon+1}{2} \right\rceil } 2^{\varepsilon } \geq 2^{-n\alpha}
    \end{equation*}
    and then
    \begin{equation}
        \label{Max_insider_bound}
        m_\text{in} \leq  \left\lceil \frac{(\log 2) \, 2^{n ( 1-h(\varepsilon/n) + \alpha )}}{N (N-1) }
        \right\rceil
    \end{equation}
\end{proof}

The proof of Lemma~\ref{lemma:bounding_intermediate_proba}.
\begin{proof}[Proof of Lemma~\ref{lemma:bounding_intermediate_proba}]
    Notice that if the $N$  $\varepsilon$-balls are disjoints, then the cardinal of the complementary of $\cup_{k=1}^{N} B_{\varepsilon (v_k)}$ is minimal.
    Then, by denoting $D_{N}$ the event that the $N$ $\varepsilon$-balls are disjoints, a lower bound of the conditional probability is:
    \begin{align*}
        \mathbb{P} \Big( \tilde v \notin \cup_{k=1}^{N}  B_{\varepsilon}(v_k)  \, \big| \, \overline{W_N} \Big) \geq \mathbb{P} \Big( \tilde v \notin \cup_{k=1}^{N}  B_{\varepsilon}(v_k) \, \big| \, D_{N}\Big)                                                          = 1 -  N V_{\varepsilon}
    \end{align*}
    that yields the upper bound.
    Concerning the lower bound, consider the complementary event:
    \begin{align*}
        1-\mathbb{P}\Big(\tilde v \notin \cup_{k=1}^{N} B_{\varepsilon}(v_k) \, \big| \, \overline{W_N}\Big) & = \mathbb{P}\Big(\tilde v \in \cup_{k=1}^{N} B_{\varepsilon}(v_k) \, \big| \, \overline{W_N}\Big)                                         \\
                                                                                                             & = \mathbb{P}\Big(\tilde v \in B_{\varepsilon}(v_1) \, \big| \, \overline{W_N}\Big)                                                        \\
                                                                                                             & \quad+\mathbb{P}\Big(\tilde v \in \cup_{k=2}^{N} B_{\varepsilon}(v_k) \, \big| \, \overline{W_N}\Big)                                     \\
                                                                                                             & \quad  - \mathbb{P}\Big(\tilde v \in B_{\varepsilon}(v_1) \, \cap \,  \cup_{k=2}^{N} B_{\varepsilon}(v_k) \, \big| \, \overline{W_N}\Big)
    \end{align*}
    and notice that:
    \begin{align*}
        \mathbb{P}\Big(\tilde v \in B_{\varepsilon}(v_1) \cap \bigcup\limits_{k=2}^{N} B_{\varepsilon}(v_k) \, \Big| \, \overline{W_N}\Big)
        \leq \sum\limits_{k=2}^{N} \mathbb{P}\Big(\tilde v \in  B_{\varepsilon}(v_1) \cap B_{\varepsilon}(v_k) \, \Big| \, \overline{W_N}\Big).
    \end{align*}
    and so on for each $v_k$ for $k=2,\dots,N-1$, then
    \begin{align*}
         & 1-\mathbb{P}\Big(\tilde v \notin \bigcup\limits_{k=1}^{N} B_{\varepsilon}(v_k) \, \Big| \, \overline{W_N}\Big) \geq \sum\limits_{k=1}^{N} \mathbb{P}\Big(\tilde v \in B_{\varepsilon}(v_k) \, \Big| \, \overline{W_N}\Big) \\
         & \qquad  - \sum\limits_{k=1}^{N-1} \sum\limits_{\ell=k+1}^{N} \mathbb{P}\Big(\tilde v \in  B_{\varepsilon}(v_k) \cap B_{\varepsilon}(v_\ell) \, \Big| \, \overline{W_N}\Big).
    \end{align*}
    Moreover, one can see that
    \begin{align*}
        \mathbb{P}\Big(\tilde v \in B_{\varepsilon}(v_k) \, \Big| \, \overline{W_N}\Big) & = \mathbb{P}\Big(\tilde v \in B_{\varepsilon}(v_k)\Big) = V_{\varepsilon}
    \end{align*}
    and
    \begin{align*}
         & \mathbb{P}\Big(\tilde v \in  B_{\varepsilon}(v_k) \cap B_{\varepsilon}(v_\ell) \, \Big| \, \overline{W_N}\Big)                                                                      \\
         & = \mathbb{P}\Big(\tilde v \in  B_{\varepsilon}(v_k) \cap B_{\varepsilon}(v_\ell) \, \Big| \, d_\mathcal{H}(v_k,v_\ell)> \varepsilon\Big)                                            \\
         & = \mathbb{P}(d_\mathcal{H}(v_k, v_\ell) > \varepsilon)^{-1} \mathbb{P}(\tilde v \in  B_{\varepsilon}(v_k) \cap B_{\varepsilon}(v_\ell), d_\mathcal{H}(v_k, v_\ell) > \varepsilon  ) \\
         & = \mathbb{P}(d_\mathcal{H}(v_k, v_\ell) > \varepsilon)^{-1}  \times                                                                                                                 \\
         & \qquad  \sum\limits_{k=\varepsilon+1}^n \mathbb{P}(\tilde v \in  B_{\varepsilon}(v_k) \cap B_{\varepsilon}(v_\ell), d_\mathcal{H}(v_k, v_\ell) = k  )                               \\
         & = \mathbb{P}(d_\mathcal{H}(v_k, v_\ell) > \varepsilon)^{-1}    \times \sum\limits_{k=\varepsilon+1}^n  \mathcal{I}_{k}^{\varepsilon} \, \mathbb{P}(d_\mathcal{H}(v_k, v_\ell) =k)
    \end{align*}
    Since $ \mathcal{I}_k^{\varepsilon} \geq \mathcal{I}_j^{\varepsilon}$ for $k \leq j$ and for any $\varepsilon$, and $\mathcal{I}_k^{\varepsilon}=0$ if $k > \varepsilon$, then
    \begin{align*}
        \mathbb{P}\Big(\tilde v \in  B_{\varepsilon}(v_k) \cap B_{\varepsilon}(v_\ell) \, \Big| \, \overline{W_N}\Big) & = \frac{\sum\limits_{k=\varepsilon+1}^{2\varepsilon} \mathcal{I}_{k}^{\varepsilon} \binom{n}{k} }{\sum\limits_{k=\varepsilon+1}^n \binom{n}{k}}              \\
                                                                                                                       & \leq \mathcal{I}_{\varepsilon+1}^{\varepsilon} \frac{V_{2\varepsilon} - V_{\varepsilon}}{1 - V_{\varepsilon}} \leq \mathcal{I}_{\varepsilon+1}^{\varepsilon}
    \end{align*}
    Then,
    \begin{align}
        1-\mathbb{P}\Big(\tilde v \notin \bigcup\limits_{k=1}^{N} B_{\varepsilon}(v_k) \, \Big| \, \overline{W_N}\Big) & \geq \sum\limits_{k=1}^{N} V_{\varepsilon} - \sum\limits_{k=1}^{N-1} \sum\limits_{\ell=k+1}^{N} \mathcal{I}^{\varepsilon}_{\varepsilon+1} \notag \\
                                                                                                                       & = N V_{\varepsilon} - \frac{N(N-1)}{2} \mathcal{I}^{\varepsilon}_{\varepsilon+1}.
        \label{eq:lemma_bound}
    \end{align}
\end{proof}

\section{Balls intersection}
\label{sec:appendix:balls_intersection}

\noindent In a recent work, \cite{durbet2022near} indicates that determining the intersection between multiple $\varepsilon-$balls in $\mathbb{F}^n_2$ corresponds to solve a specific linear system $L$, which depends on the template database.
First, we provide Proposition~\ref{prop:disjoint_balls} that indicates the probability that there is none intersection between balls, or in other words that all the balls are disjoint sets.
\begin{proposition}
    \label{prop:disjoint_balls}
    For a template database $\mathcal{B} = (v_1, \dots, v_N)$, $D_N$ is the event "the balls $B_\varepsilon(v_k)$, for $k \in \left\{ 1, \dots, N\right\}$ are disjoint sets in $\mathbb{F}^n_2$".
    Then,
    \begin{align*}
        \prod\limits_{j=1}^N \big(1- (j-1)V_{2\varepsilon}  \big) \leq  \mathbb{P}(D_N ) \\
        \prod\limits_{j=2}^N \left(1- (j-1 ) V_{2\varepsilon} + \frac{(j-1)(j-2)}{2}\mathcal{I}^{2\varepsilon}_{2\varepsilon+1} \right) \ge \mathbb{P}(D_N ).
    \end{align*}
\end{proposition}
\begin{proof}
    Notice that two balls $B_\varepsilon(v_1)$ and $B_\varepsilon(v_2)$ are disjoint sets iff $d_\mathcal{H}(v_1,v_2) > 2\varepsilon$, which can be formulated as: $v_2 \notin B_{2\varepsilon}(v_1)$.
    By using chain rule and where $\mathcal{B} = \{ v_1, \dots, v_N \}$:
    \begin{align}
        \label{eq:DN_chain_rule}
        \mathbb{P}(D_N) =  \prod_{j=2}^N \mathbb{P}\Big(v_j \notin \bigcup\limits_{k=1}^{j-1}  B_{2\varepsilon}(v_k)  \, \Big| \,  \overline{W_{j-1}} \Big)
    \end{align}
    where $\overline{W_{j-1}}$ denotes the event "$v_2 \notin B_{2\varepsilon}(v_1) , \dots, v_{j-1} \notin \bigcup\limits_{k=1}^{j-2}  B_{2\varepsilon}(v_k)$".
    Then, according to Lemma~\ref{lemma:bounding_intermediate_proba}, each term of Equation~\eqref{eq:DN_chain_rule} can be bounded by below and above.
    For instance, the $j^\text{th}$ one gives:
    \begin{align*}
        \mathbb{P}\left(v_j \notin \cup_{k=1}^{j-1}  B_{2\varepsilon}(v_k)  \, \big| \,   \overline{W_{j-1}} \right) & \ge 1- (j-1)V_{2\varepsilon}                                             \\
        \mathbb{P}\left(v_j \notin \cup_{k=1}^{j-1}  B_{2\varepsilon}(v_k)  \, \big| \,   \overline{W_{j-1}} \right) & \leq 1- (j-1 ) V_{2\varepsilon}                                          \\
                                                                                                                     & \quad + \mathcal{I}^{2\varepsilon}_{2\varepsilon+1}\frac{(j-1)(j-2)}{2}.
    \end{align*}
\end{proof}

In the following, $d_S$ denotes the Hamming distance for a subset $S \subset \{1,\dots,n\}$, and $S = \{ S_1, \dots, S_K\}$ a partition of $\{1,\dots,n\}$.
Moreover, an equivalence relation $\sim_\mathcal{B}$ is defined on $\{1,\dots,n\}$ depending on the template database $\mathcal{B}$:
\begin{equation}
    k \sim_\mathcal{B} j \quad \text{ if } C^k = C^j \text{ or } C^k = 1 - C^j,
\end{equation}
where $C^k$ is the $k^\text{th}$ column of the template database, in other words $C^k_j$ is the value of the $j^\text{th}$ template $v_j$ on the $k^\text{th}$ coordinate of $\mathbb{F}^n_2$.
\begin{definition}[Partitioning according to $\sim_\mathcal{B}$]
    For a given template database $\mathcal{B}$, a partition $S = (S_1, \dots, S_K)$ of $\{1,\dots,n\}$ is in accordance with $\sim_\mathcal{B}$ if each $C_k$ is an equivalence class of $\sim_\mathcal{B}$.
\end{definition}
Consider $I = (I_1, \dots, I_{K_0})$ denoting the smallest partition of $\{1,\dots,n\}$ obtained from the $\sim_\mathcal{B}$ relation.
Theorem~\ref{thm:durbet2022near:intersection_linear_system} from~\cite{durbet2022near} is recalled below.
\begin{theorem}[Balls intersection and linear system]
    \label{thm:durbet2022near:intersection_linear_system}
    For a given template database $\mathcal{B}$, and a given $\varepsilon$, and an arbitrary reference template $v_0 \in \mathcal{B}$, $p \in \cap_{v \in \mathcal{B}} B_\varepsilon(v) \Leftrightarrow A P \leq e_{v_0}$
    where $ p \in \mathbb{F}^n_2$, $e_{v_0} = (\varepsilon- d_\mathcal{H}(v_1,v_0), \dots, \varepsilon - d_\mathcal{H}(v_N,v_0))$ and $A$ is an $N \times |I|$ matrix whose the $(i,j)^\text{th}$ element is:
    \begin{equation*}
        a_{i,j} = \begin{cases}
            1  & \text{si } d_{I_j}(v_0,v_i) = 0     \\
            -1 & \text{si } d_{I_j}(v_0,v_i) = |I_j| \\
        \end{cases}
        \vspace*{-0.2cm}
    \end{equation*}
    where $P$ is a vector whose the $k^\text{th}$ element is $P_k = d_{I_k} (p,v_0)$.
\end{theorem}
In other words, looking for an element $p$ belonging to the intersection of $N$ $\varepsilon$-balls is equivalent to solve a specific linear system with a vector parameter $P$ that belongs to the discrete space $\mathcal{P} = \prod_{k=1}^{K_0} \big\{ 0, \dots, \min(\varepsilon, |I_k|)\big\}$.
In order to provide the following result, it is required to introduce  $\ell_{\varepsilon} = \{P \in \mathcal{P} \, | \, AP \leq e_{v_0}\}$ for a given template database $\mathcal{B}$. Then the result is:
\begin{corollary}
    \label{cor:cardinal:intersection}
    The cardinal of the intersection of $\varepsilon-$balls centered on $v_1, \dots, v_N$ is
    \begin{equation*}
        \left| \bigcap\limits_{j=1}^N B_\varepsilon (v_j) \right| = \sum\limits_{(P_1, \dots, P_K) \in \ell_\varepsilon} \prod\limits_{k=1}^{K_0} \binom{|I_k|}{P_k}.
    \end{equation*}
\end{corollary}
\begin{proof}
    In order to determine the targeted cardinal the proof below relies on counting the number of elements $p \in \mathbb{F}_2^n$ which are associated to each $P \in \ell$.
    Since $P_k=  d_{I_k}(p,v_0)$, then $0 \leq P_k \leq |I_k|$.
    Moreover, it does not matter which are the $P_k$ coordinates among the subset $I_k$ for which $p$ is different from $v_0$.
    Then, there is $\binom{|I_k|}{P_k}$ possibilities, and so on for the other subset $I_{k'}$, therefore for $P \in \mathcal{P}$, there are $\prod_{k=1}^{K_0} \binom{|I_k|}{P_k}$ associated elements in $\mathbb{F}^n_2$.
    According to Theorem~\ref{thm:durbet2022near:intersection_linear_system}, each $p$ in the intersection is associated to a vector $P \in \ell_\varepsilon$, and each vector $P$ corresponds to only elements $p$ in the intersection, thus it is sufficient to find the result.
\end{proof}

Furthermore, a result about the relative size of $\mathcal{P}$ indicating the gain to work in $\mathcal{P}$ to find an element $p$ in the intersection of $N$ $\varepsilon-$balls, rather than working in $\mathbb{F}_2^n$ is provided.
\begin{proposition}[Cardinal reduction]
    \label{prop:cardinalreduction}
    For a given database $\mathcal{B}$ of size $N$, the cardinal of $\mathcal{P}$ is at least lower than $2^n$ and in particular:
    \begin{equation*}
        \frac{|\mathbb{F}^n_2|}{ |\mathcal{P}|} \geq \prod\limits_{k=1}^{K_0} \frac{2^{|I_k|}}{|I_k| + 1}.
    \end{equation*}
\end{proposition}
\begin{proof}
    First, notice that $\sum_{k=1}^{K_0} |I_k| = n$, and then
    \begin{align*}
        \frac{|\mathbb{F}_2^n|}{|\mathcal{P}|} & \geq \frac{2^{\sum_{k=1}^{K_0} |I_k|}}{\prod_{k=1}^{K_0} (|I_k| +1)}  = \prod\limits_{k=1}^{K_0} \frac{2^{|I_k|}}{|I_k|+1}.
    \end{align*}
\end{proof}
In order to give useful insight, the following result indicated what is the size of the partition $I$, that drives the cardinal of $\mathcal{P}$ and then the complexity of linear system $L$:
\begin{theorem}
    \label{thm:distribution:partition}
    For a template database $\mathcal{B} \subset \mathbb{F}^n_2$ of size $N$, the distribution of the partition size is given by
    \begin{equation*}
        \mathbb{P}\big( |I|= i\big) = \frac{1}{2^{(N-1)(n-i)}} \sum\limits_{(n_1, \dots,n_i) \in S_{n-i}^i} \prod\limits_{j=1}^{i} j^{n_j} \left(1- \frac{j-1}{2^{{N-1}}}\right)
    \end{equation*}
    for $i=1, \dots, n$ and where $S_n^i = \left\{ (n_1, \dots, n_i) : n_k \geq 0, \sum_{j=1}^i n_j = n \right\}$.
\end{theorem}
\begin{proof}
    In order to write the event "$|I| = i$", proceed by working column by column of the template database $\mathcal{B}$, and in particular by denoting $|I_k|$ the number of equivalent classes of $\sim_\mathcal{B}$ among the $k$ first columns $C^1, \dots, C^k$.
    Next, two conditional probability properties are introduced.
    The first one indicates the probability to not increase the number of equivalent classes by appending a new column.
    Since each equivalent class contains only two columns (elements of $\mathbb{F}^N_2$):

    \begin{align}
        \label{eq:In_conditional_1}
        \mathbb{P}\Big( |I_{k+1}| = j \, \Big| \,  |I_k| = j\Big) & = \frac{2j}{2^N} = \frac{j}{2^{N-1}}
    \end{align}
    The second property corresponds to the case of obtaining a new equivalent class by appending a new column:
    \begin{align}
        \label{eq:In_conditional_2}
        \mathbb{P}\Big( |I_{k+1}| = j+1 \, \Big| \,  |I_k| = j\Big) & = \frac{2^N - 2j}{2^N} = 1 - \frac{j}{2^{N-1}}.
    \end{align}
    Then, by using chain rule several times:
    \begin{align*}
         & \mathbb{P}\Big( |I_{n}| = i  \Big)                                                                                = \mathbb{P}\Big( |I_{n}| = i \, \Big| \, |I_{n-1}| = i \Big) \times  \mathbb{P}\Big( |I_{n-1}| = i \Big) \\
         & \quad + \mathbb{P}\Big( |I_{n}| = i \, \Big| \, |I_{n-1}| = i-1 \Big) \times  \mathbb{P}\Big( |I_{n-1}| = i-1 \Big)
    \end{align*}
    and so on until obtaining that $\mathbb{P}( |I_{n}| = i  )$ is only expressed with \eqref{eq:In_conditional_1}, \eqref{eq:In_conditional_2} and $\mathbb{P}( |I_{1}| = 1 ) =1$.
    The probability to compute is then a sum of product of sequential conditional probabilities, and each sequence can be written as:
    \begin{align*}
        \prod_{k=1}^{n-1} \mathbb{P}\Big( |I_{k+1}| = x_{k+1} \, \Big| \, |I_{k}| = x_{k} \Big)
    \end{align*}
    where $x_{k+1} \in \{x_{k},x_{k}+1\}, \, x_1 = 1$  and $x_n = i$.
    Consider, for $j \in \{1, \dots, i\}$ the quantity "$n_j+1$", the number of elements $x_k = j$, for $k \in \{1, \dots,n\}$, so that $n_j$ corresponds to the number of times that appending a new column does not increase the number of equivalent classes (associated to the conditional probability \eqref{eq:In_conditional_1}).
    Then, each product of sequential conditional probabilities can be written as:
    \begin{align*}
        \prod\limits_{k=1}^{n-1} \mathbb{P}\Big( |I_{k+1}| = x_{k+1} \, \Big| \, |I_{k}| = & x_{k} \Big)  = \prod\limits_{j=1}^i  \left( 1 - \frac{j}{2^{N-1}} \right)  \left( \frac{j}{2^{N-1}} \right)^{n_j} \\
                                                                                           & = \frac{1}{2^{(N-1)(n-i)}}  \prod\limits_{j=1}^i j^{n_j} \left( 1 - \frac{j}{2^{N-1}} \right)
    \end{align*}
    since $\sum_{j=1}^i (n_j +1) = n$.
    The results is then obtained by summing on each sequence, which reduces to summing on all possible vector $(n_1, \dots, n_i)$ belonging to $S^i_{n-i}$, a simplex-like set.
\end{proof}
Moreover, since computing a sum on elements belonging to sets such as $S^i_n$ can be difficult and time-consuming in practice, below are also provided bounds to work with probabilities like $\mathbb{P}(|I| = i)$.
\begin{corollary}
    \label{cor:boundpartition}
    For a template database $\mathcal{B} \subset \mathbb{F}^n_2$ of size $N$, the distribution of the partition size is bounded by
    \begin{align*}
        {n \brace i} \frac{1}{2^{(N-1)(n-i)}} \prod\limits_{j=1}^{i}  \left( 1 - \frac{j-1}{2^{N-1}} \right)       & \leq  \mathbb{P}\big( |I|= i\big) \\
        {n \brace i} \frac{i^{n-i}}{2^{(N-1)(n-i)}} \prod\limits_{j=1}^{i}  \left( 1 - \frac{j-1}{2^{N-1}} \right) & \ge \mathbb{P}\big( |I|= i\big)   \\
    \end{align*}
    for $i=1, \dots, n$ and where $ {n \brace i}$ is the Stirling number of the second kind.
\end{corollary}
\begin{proof}
    First, the product $\prod_{j=1}^i j^{n_j}$ is bounded, according that $(n_1, \dots, n_i) \in S_{n-i}^i$.
    In particular, for $(n-i, 0, \dots, 0)$ and $(0, \dots,0,n-i)$, which belong to $S_{n-i}^i$, the lower and the upper bounds are respectively obtained as $1 \leq \prod_{j=1}^i j^{n_j} \leq i^{n-i}$.
    The results follows by observing that $S_{n-i}^i$ is equivalent to the set of possible partitions of $n$ objects into $k$ subsets, and then its cardinal is given with a Stirling number of the second kind.
\end{proof}
Some last results are given below about balls intersection, in order to bound $\mathcal{I}_d^\varepsilon$ and to emphasize the link with the binary entropy function $h$.
First, the proof of an important lemma in section~\ref{utg_attacks:in} (Lemma~\ref{lemma:entropyballintersection}) is given below.
\begin{proof}[Proof of Lemma~\ref{lemma:entropyballintersection}]
    Use Equation~\eqref{eq:approximation_binomial_sum} to obtain the following upper bound:
    \begin{align*}
        \sum\limits_{i=0}^{\min(\varepsilon-k, \varepsilon-d+k)} \binom{n-d}{i} \leq 2^{n h \left( \frac{\min(\varepsilon-k, \varepsilon-d+k)}{n}\right)}
    \end{align*}
    Then, notice that if $0 \leq d \leq \varepsilon$ then, $\min(\varepsilon-k, \varepsilon-d+k) \leq \varepsilon-d + \lfloor d/2 \rfloor$, for any $k \in \left\{ \max(0, d-\varepsilon), \dots, \min(\varepsilon,d)\right\}$, and if $\varepsilon < d \leq  2 \varepsilon$ then, $\min(\varepsilon-k, \varepsilon-d+k) \leq \varepsilon-\lceil d/2 \rceil$.
    Since $d - \lfloor d/2 \rfloor = \lceil d/2 \rceil$, then
    \begin{align*}
        \mathcal{I}_d^\varepsilon & \leq  \frac{1}{2^n}2^{nh\left(\frac{\varepsilon-\lceil d/2 \rceil}{n}\right)} \sum\limits_{k=\max(0, d-\varepsilon)}^{\min(\varepsilon,d)} \binom{d}{k} \leq \frac{1}{2^n} 2^{nh\left(\frac{\varepsilon-\lceil d/2 \rceil}{n}\right)+d}
    \end{align*}
    To obtain the second upper bound, consider the following Taylor linear approximation.
    Since $h$ is a concave function, it follows that
    \begin{align*}
        h\left( \frac{\varepsilon - \left\lceil d/2 \right\rceil}{n}\right) & \leq  h(\varepsilon/n) + \frac{\left\lceil d/2 \right\rceil}{n} \log_2 \left( \frac{\varepsilon/n}{1-\varepsilon/n} \right)
    \end{align*}
\end{proof}
\begin{lemma}
    \label{lemma:entropyballintersection_min}
    For $u$ and $v \in \mathbb{F}^n_2$, $d = d_\mathcal{H}(u,v)$, with $d \leq 2\varepsilon$, $\varepsilon/n < 1/2$, then
    \begin{align*}
        \mathcal{I}_d^\varepsilon & \geq 2^{(-n+d)\left( 1-h\left(\frac{ \varepsilon-d}{n-d}\right) \right) } \Big/ \sqrt{8(\varepsilon-d)\left(\frac{ n-\varepsilon}{n-d}\right)}               \text{ if $d<\varepsilon$, }                                \\
        \mathcal{I}_d^\varepsilon & \geq  2^{-n + \varepsilon}\left(1-2^{h\left(\frac{d-\varepsilon}{\varepsilon}\right)}\right)                                                                 \text{ if $\varepsilon \leq d < \frac{3}{2} \varepsilon$, } \\
        \mathcal{I}_d^\varepsilon & \geq 2^{-n+\varepsilon h\left( \frac{2\varepsilon-d}{\varepsilon} \right)} \Big/ \sqrt{8(2\varepsilon-d)\left(1- \frac{2\varepsilon-d}{\varepsilon}\right)}                                                              \\
                                  & \qquad \text{ if $\frac{3}{2}\varepsilon < d \leq 2 \varepsilon$.}
    \end{align*}
\end{lemma}
\begin{proof}
    First, consider the case "$d< \varepsilon$", and observe that $\operatorname{min}(\varepsilon-k,\varepsilon-d+k) \geq \varepsilon-d$ for at least one of the value $k \in \left\{0, \dots, d \right\}$.
    Then,
    \begin{equation*}
        \sum\limits_{i=0}^{\operatorname{min}(\varepsilon-k,\varepsilon-d+k)} \binom{n-d}{i} \geq \sum\limits_{i=0}^{\varepsilon-d} \binom{n-d}{i} \geq \frac{2^{(n-d) h\left(\frac{ \varepsilon-d}{n-d} \right)}}{ \sqrt{8(\varepsilon-d)(1-\frac{ \varepsilon-d}{n-d})}}
    \end{equation*}
    since $(\varepsilon-d)/(n-d) < 1/2$ if $\varepsilon/n < 1/2$.
    Moreover $\operatorname{max}(0,d-\varepsilon) = 0$ so that
    \begin{equation*}
        \sum\limits_{k=\max(0, d-\varepsilon)}^{\min(\varepsilon,d)} \binom{d}{k} = \sum\limits_{k=0}^{d} \binom{d}{k} = 2^d
    \end{equation*}

    For both other cases, $d\geq \varepsilon$
    then, $\operatorname{min}(\varepsilon-k,\varepsilon-d+k) = 0$, for at least one of the value $k \in \left\{ d-\varepsilon,\dots, \varepsilon \right\}$ and then
    \begin{equation*}
        \sum\limits_{i=0}^{\operatorname{min}(\varepsilon-k,\varepsilon-d+k)} \binom{n-d}{i} \geq \binom{n-d}{0}= 1
    \end{equation*}
    Next, $\operatorname{max}(0,d-\varepsilon) = d-\varepsilon$, and if $\varepsilon \leq d < \frac{3}{2} \varepsilon$ then  $d-\varepsilon < \varepsilon/2$ and
    \begin{align*}
        \sum\limits_{k=\max(0, d-\varepsilon)}^{\min(\varepsilon,d)} \binom{d}{k} & \geq \sum\limits_{k=d-\varepsilon}^{\varepsilon} \binom{\varepsilon}{k}                = \sum_{k=0}^\varepsilon \binom{\varepsilon}{k} -\sum_{k=0}^{d-\varepsilon} \binom{\varepsilon}{k} \\
                                                                                  & \geq 2^\varepsilon - 2^{\varepsilon h\left( \frac{d-\varepsilon}{\varepsilon}\right)}
    \end{align*}
    If $\frac{3}{2} \varepsilon < d \leq 2\varepsilon$, then $2\varepsilon-d < \varepsilon/2$ and
    \begin{align*}
        \sum\limits_{k=\max(0, d-\varepsilon)}^{\min(\varepsilon,d)} \binom{d}{k} & \geq \sum\limits_{k=d-\varepsilon}^{\varepsilon} \binom{\varepsilon}{k} = \sum_{k=0}^{2\varepsilon-d} \binom{\varepsilon}{k}
        \\ & \geq \frac{2^{\varepsilon h\left( \frac{2\varepsilon-d}{\varepsilon} \right)}}{\sqrt{8(2\varepsilon-d)\left(1- \frac{2\varepsilon-d}{\varepsilon}\right)}}
    \end{align*}
\end{proof}
\begin{lemma}
    \label{lemma:entropyballintersection_min_epsilon_p_1}
    For $u$ and $v \in \mathbb{F}^n_2$, $d_\mathcal{H}(u,v)=\varepsilon+1$, with $\varepsilon/n < 1/2$, then $\mathcal{I}_{\varepsilon+1}^\varepsilon \geq 2^{-n}(2^{\varepsilon}-1)$.
\end{lemma}
\begin{proof}
    Notice that $\operatorname{min}(\varepsilon-k, k-1) = 0$ if $k=1$ or $k = \varepsilon$, then
    \begin{equation*}
        \sum\limits_{i=0}^{\operatorname{min}(\varepsilon-k, k-1)} \binom{n-\varepsilon-1}{i} \geq 1
    \end{equation*}
    Next
    \begin{equation*}
        \sum\limits_{k=1}^\varepsilon \binom{\varepsilon +1 }{k}  \geq  \sum\limits_{k=1}^\varepsilon \binom{\varepsilon }{k} =  \sum\limits_{k=0}^\varepsilon \binom{\varepsilon }{k}  - 1 = 2^\varepsilon-1
    \end{equation*}
\end{proof}

\begin{table*}[]
    \centering
    \begin{tabular}{c|cc|cc|cc|cc|cc}
        \toprule
        \multicolumn{1}{c|}{$n$}                                                                                            & $128$                        & \multicolumn{1}{c|}{$256$}                        & \multicolumn{2}{c|}{$128$}                        & \multicolumn{2}{c|}{$128$}   & \multicolumn{2}{c|}{$128$}    & \multicolumn{2}{c}{$128$}                                    \\ \midrule
        \multicolumn{1}{c|}{$\varepsilon$}                                                                                  & \multicolumn{2}{c|}{$30$}    & $10$                                              & \multicolumn{1}{c|}{$20$}                         & \multicolumn{2}{c|}{$30$}    & \multicolumn{2}{c|}{$30$}     & \multicolumn{2}{c}{$30$}                                     \\ \midrule
        \multicolumn{1}{c|}{$N$ ($\log_{10}$)}                                                                              & \multicolumn{2}{c|}{$6$}     & \multicolumn{2}{c|}{$6$}                          & $5 $                                              & \multicolumn{1}{c|}{$7 $}    & \multicolumn{2}{c|}{$6$}      & \multicolumn{2}{c}{$6$}                                      \\ \midrule
        \multicolumn{1}{c|}{$a$ ($\log_{10}$)}                                                                              & \multicolumn{2}{c|}{$3$}     & \multicolumn{2}{c|}{$3$}                          & \multicolumn{2}{c|}{$3$}                          & $2$                          & \multicolumn{1}{c|}{$4$}      & \multicolumn{2}{c}{$3$}                                      \\ \midrule
        \multicolumn{1}{c|}{$\kappa$ ($\log_2$)}                                                                            & \multicolumn{2}{c|}{$47$}    & \multicolumn{2}{c|}{$47$}                         & \multicolumn{2}{c|}{$47$}                         & \multicolumn{2}{c|}{$47$}    & $20$                          & $100$                                                        \\ \midrule
        \multicolumn{1}{c|}{$\frac{\mathbb{P} \left( A_h(\kappa) \leq a \right)}{\mathbb{P} \left( A_h(0) \leq a \right)}$} & \multicolumn{1}{c}{$0.9981$} & \multicolumn{1}{c|}{$\phantom{0.0}1\phantom{00}$} & \multicolumn{2}{c|}{$\phantom{0.0}1\phantom{00}$} & \multicolumn{1}{c}{$0.9998$} & \multicolumn{1}{c|}{$0.9808$} & \multicolumn{2}{c|}{$0.9981$} & \multicolumn{2}{c}{$0.9981$} \\
        \bottomrule
    \end{tabular}
    \caption{Ratio between a $\kappa$-adaptive attacker and a $0$-adaptive attacker in function of $n$, $\varepsilon$, $N$, $a$ and $\kappa$ for the lower bound.}
    \label{Tab:equiv_ratio}
\end{table*}

\section{Multi-NC and Master-templates}
\label{mlty_near_collisions}

The probability of a weak collision is an important starting point to determine the probability that there exists a template that covers all others. As the multi-near-collision states that there exists in a database a template which impersonates one or several templates, it is a particular case of master template.

\begin{definition}[$(k,\varepsilon)$-master template]
    A $(k,\varepsilon)$-master template $t$ with respect to $\mathcal{B}$ is a template such that there exists a $k$-tuple of distinct templates $a_1,\dots,a_k$ in $\mathcal{B}$ with $d_\mathcal{H}(a_i,t) \leq \varepsilon$ for all  $i \in \lbrace 1,\dots,k \rbrace$.
\end{definition}

To compute the probability that there exists a $N$-master template, we firstly introduce  $\mathcal{C}(\varepsilon,N)$ the number of template databases that belong to a given $\varepsilon-$ball,
\begin{equation*}
    \mathcal{C}(\varepsilon,N) = \left\{ \mathcal{B} = (v_1, \dots, v_N) \in \left(\mathbb{Z}_2^n \right)^N,  \, \mathcal{B} \subset B_\varepsilon \right\},
\end{equation*}
and $\mathcal{C}_v(\varepsilon,N)$ is the same number but for a $\varepsilon-$ball centered on a given $v$.
Then, the following theorem provides the probability that there exists a $(N,\varepsilon)$-master template.
\begin{theorem}
    \label{thm:proba:NMT}
    For a template database $\mathcal{B}(v_1,\dots, v_N)$, the probability that there exists a $(N,\varepsilon)$-master template is:

    $$ \mathbb{P}(\mathcal{B} \subset B_\varepsilon) = \frac{1}{2^{n(N-1)}} \sum\limits_{B \in \mathcal{C}(\varepsilon,N)} \left| B_\varepsilon^\cap(B)  \right|^{-1}$$
    with $B_\varepsilon^{\cap}(B)= \bigcap_{i=1}^N B_\varepsilon (v_i).$
    Lower and upper bounds for this probability are $
        V_\varepsilon^{N-1} \leq \mathbb{P}(\mathcal{B} \subset B_\varepsilon) \leq V_{2\varepsilon}^{N-1}$,
    where $V_\varepsilon$ is the measure of an $\varepsilon$-ball.
\end{theorem}

\begin{proof}
    First, in order to introduce the intuition to obtain the first result, observe that the proportion of template databases which can be covered by an $\varepsilon-$ball is bounded from above by counting the number of template databases included in each balls of $\mathbb{Z}^n_2$:
    \begin{equation*}
        \mathbb{P}(\mathcal{B} \subset B_\varepsilon) \leq \frac{1}{2^{Nn}} \sum_{v \in \mathbb{Z}^n_2} \left| \mathcal{C}_v(\varepsilon,N)\right|
    \end{equation*}

    This quantity is an upper bound since each template database $B$ is counted several times, for different $v$.
    To be more specific, a template database is counted once for each $v \in B_\varepsilon^{\cap}(B)$.
    Next, notice that for a given database $B \in \mathcal{C}_v(\varepsilon,N)$, $\left| B_\varepsilon^{\cap}(B) \right| \geq 1$, since $v$ necessarily belongs to $B_\varepsilon^{\cap}(B)$.
    Then, the following equation can be written:
    \begin{equation*}
        \mathbb{P}(\mathcal{B} \subset B_\varepsilon) = \frac{1}{2^{Nn}}\sum\limits_{v\in \mathbb{Z}_2^n} \sum\limits_{B \in \mathcal{C}_v(\varepsilon,N)} \left| B_\varepsilon^\cap(B)  \right|^{-1}.
    \end{equation*}
    Observe that the sum on $\mathcal{C}_v(\varepsilon,N)$ does not depend on $v$, so that:
    \begin{align*}
        \mathbb{P}(\mathcal{B} \subset B_\varepsilon) & = \frac{1}{2^{Nn}}\sum\limits_{v\in \mathbb{Z}_2^n} \sum\limits_{B \in \mathcal{C}(\varepsilon,N)} \left| B_\varepsilon^\cap(B)  \right|^{-1} \\
                                                      & = \frac{1}{2^{n(N-1)}} \sum\limits_{B \in \mathcal{C}(\varepsilon,N)} \left| B_\varepsilon^\cap(B)  \right|^{-1}
    \end{align*}
    Next, to obtain the lower and upper bounds consider:
    \begin{align*}
        I_\varepsilon(v_1, \dots, v_k)  = \cup_{v \in \cap_{i=1}^k B_\varepsilon(v_i)} B_\varepsilon(v)
    \end{align*}
    and it follows that:
    \begin{align*}
        \mathbb{P}(\mathcal{B} \subset B_\varepsilon)  = & \, \mathbb{P}\Big( v_2 \in B_{2\varepsilon}(v_1) \, \Big| \, v_1\Big)                                        \\
                                                         & \quad\times \mathbb{P}\Big(v_3 \in I_\varepsilon(v_1, v_2) \, \Big| \, (v_1, v_2) \subset B_\varepsilon\Big) \\
                                                         & \quad \times \dots                                                                                           \\
                                                         & \quad \times \mathbb{P}\Big(v_N \in I_\varepsilon(v_1, \dots, v_{N-1}) \,                                    \\
                                                         & \quad\quad \Big| \, (v_1, \dots, v_{N-1}) \subset B_\varepsilon\Big)
    \end{align*}
    Next, provide an upper bound for each previous term:
    \begin{align}
        V_\varepsilon \leq \mathbb{P}\Big(v_k \in I_\varepsilon(v_1, \dots, v_{k-1}) \, \Big| \, (v_1, \dots, v_{k-1}) \subset B_\varepsilon\Big) \leq V_{2\varepsilon}.
    \end{align}
    The upper bound is obtained by considering the case "$v_1 = \dots = v_{k-1}$".
    Concerning the lower bound, it is based on the fact that "$(v_1, \dots, v_{k-1}) \subset B_\varepsilon$" implies that $\cap_{i=1}^{k-1} B_\varepsilon(v_i)$ is a non-empty set, then $\big| I_\varepsilon(v_1, \dots, v_{k-1}) \big| \geq |B_\varepsilon|$.
\end{proof}

Theorem~\ref{thm:proba:NMT} gives the probability that there exists a template in $\mathbb{Z}_2^n$ which impersonates all users if the threshold of $\varepsilon$ is used.
This results only refers to a rare event, but it is an intermediate result in order to provide Corollary~\ref{cor:proba:kMT}, which focuses on the probability of a $(k,\varepsilon)$-master template.
In the following, $\mathcal{B}_k$ denotes a subset of $k$ templates from the template database $\mathcal{B}$, and $M_k$ is the event "an $\varepsilon-$ball covers $k$ templates and none ball can not include these $k$ templates plus another template from the template database".
\begin{corollary}
    \label{cor:proba:kMT}
    The probability of a $(k,\varepsilon)$-master template is:
    \begin{equation*}
        \mathbb{P}(M_k) =  \mathbb{P}(\mathcal{B}_k \subset B_\varepsilon) \times \frac{\binom{N}{k}}{2^{n(N-k)}} \left| \bigcap\limits_{v \in B_\varepsilon^\cap(\mathcal{B}_k)} \overline{B_\varepsilon(v)} \right|^{N-k}
    \end{equation*}
    and if $\mathcal{B}_k \subset B_\varepsilon$, a lower bound is $ \left( 1 - V_{2\varepsilon}\right)^{N-k} $ and an upper bound is $\left( 1 - V_\varepsilon\right)^{N-k}$,
    where $V_\varepsilon$ is the measure of an $\varepsilon$-ball.
\end{corollary}

\begin{proof}
    A $(k,\varepsilon)$-master template occurs when $k$ templates can be covered with one (or more) $\varepsilon-$ball and the $N-k$ other templates do not belong to the covering ball(s):
    \begin{align*}
        \mathbb{P}(M_k) & = \mathbb{P}(\mathcal{B}_k \subset B_\varepsilon)
        \\ & \phantom{\times} \times  \mathbb{P} \Big( \forall u \in \mathcal{B} \backslash \mathcal{B}_k, u \notin \bigcup\limits_{v \in B_\varepsilon^\cap(\mathcal{B}_k) }B_\varepsilon(v) \, \Big| \, \mathcal{B}_k \subset B_\varepsilon\Big) \\
                        & = \binom{N}{k}\mathbb{P}\big( (v_1, \dots, v_k) \subset B_\varepsilon\big)                                                                                                                                 \\
                        & \phantom{\times}\times \prod\limits_{j=1}^{N-k}\mathbb{P} \Big(\bigcap\limits_{v \in B_\varepsilon^\cap(\mathcal{B}_k) } \overline{B_\varepsilon(v)} \, \Big| \, \mathcal{B}_k \subset B_\varepsilon \Big) \\
                        & = \binom{N}{k}\mathbb{P}\big( (v_1, \dots, v_k) \subset B_\varepsilon\big)                                                                                                                                 \\
                        & \phantom{\times} \times \frac{1}{2^{n(N-k)}} \left| \bigcap\limits_{v \in B_\varepsilon^\cap(\mathcal{B}_k)} \overline{B_\varepsilon(v)} \right|^{N-k}
    \end{align*}

    Next, in order to obtain lower and upper bounds, remark that, if $\mathcal{B}_k \subset B_\varepsilon$ holds, $B_\varepsilon^\cap(\mathcal{B}_k)$ is not empty, and that
    \begin{equation*}
        \{v\} \subseteq B_\varepsilon^\cap(\mathcal{B}_k)  \subseteq B_{\varepsilon},
    \end{equation*}
    for an unknown specific $v \in \mathbb{Z}^n_2$.
    The lower case occurs if at least two templates $u$ and $v \in \mathcal{B}_k$ are such that $d_\mathcal{H}(u,v) =2\varepsilon$ (they are opposed on an $\varepsilon$-sphere)  and the upper case occurs when all templates of $\mathcal{B}_k$ are equal.
    The results follow by replacing $B_\varepsilon^\cap(\mathcal{B}_k)$ with $\{v\}$ and then with $B_{\varepsilon}$.
\end{proof}

To a lesser extent, the probability of a $(k,\varepsilon)$-near-collision with a particular template is given in Proposition~\ref{prop:proba_near_collision} below.
\begin{proposition}
    \label{prop:proba_near_collision}
    For a given template $v \in \mathcal{B}$, and $\mathcal{B}_{-v} = \mathcal{B} \backslash v$, the probability of a near-collision for $v$ is:
    \begin{align*}
        \mathbb{P}\big(\exists (v_1,\dots, v_k) \subset \mathcal{B}_{-v} \text{ such that } (v_1,\dots, v_k) \subset B_\varepsilon(v)\big) \\ = \binom{N-1}{k} V_\varepsilon^k \left(1- V_\varepsilon \right)^{N-k-1}
    \end{align*}
\end{proposition}

\begin{proof}
    Recall that each $u \in \mathcal{B}_{-v}$ are independent and follows a uniform distribution on $\mathbb{Z}^n_2$.
    Denotes $v = (v_1,\dots, v_k)$ a vector in $\mathcal{B}_{-v}$, then $\mathbb{P}\big(\exists v \subset \mathcal{B}_{-v} \text{ such that } v \subset B_\varepsilon(v)\big)$ is equal to
    \begin{align*}
        \binom{N-1}{k} \prod\limits_{v_j \in v} \mathbb{P}(v_j \in B_\varepsilon(v)) \prod\limits_{v_\ell \notin v} \mathbb{P}(v_\ell \notin B_\varepsilon(v))
    \end{align*}
    and the result follows with $\mathbb{P}(v_j \in B_\varepsilon(v)) = V_\varepsilon$.
\end{proof}
Proposition~\ref{prop:proba_near_collision} can be used to compute the possible
values of $\FMR$ and the associated probabilities after the registration of each
client as shown in Section~\ref{template_sec}.

\section{Numerical Results: Proposition~\ref{prop:outsider_attack:negligible_difference_range}}

To support Proposition~\ref{prop:outsider_attack:negligible_difference_range} considering finite parameters, Table~\ref{Tab:equiv_ratio} investigates several settings by computing the ratio $$\frac{\mathbb{P} \left( A(\kappa) \leq a \right)}{\mathbb{P} \left( A(0) \leq a \right)}$$ as well as $\mathbb{P} \left( A(\kappa) \leq a \right)$.
In the following, the union size is maximized (\textit{i.e.}, all the balls are disjoint). Hence, a case in favour of the attacker and against our proposal is investigated.
One of our realistic settings is defined as $(n\ge 128,\varepsilon\leq 0.25n ,N=10^6, \kappa \leq |B_\varepsilon|, a=10^3)$ where $a$ and $\kappa$ are respectively the number of trials and the extra information of a $\kappa$-adaptive attacker.
In this case, we observe a ratio very close to $1$.
It is interesting to note that even with a large increase in the information given to the $\kappa$-adaptive attacker, her probability of success does not increase significantly.
Hence, as shown by Proposition~\ref{prop:outsider_attack:negligible_difference_range}, Proposition~\ref{prop:outsider_attack:negligible_difference}, Theorem~\ref{cor:naive_eq_k_adap:median_and_expect} and Table~\ref{Tab:equiv_ratio}, a $0$-adaptive attacker performs as well as a $\kappa$-adaptive attacker for reasonable parameters.

\end{document}